\documentclass[sn-nature]{sn-jnl}

\usepackage{graphicx}%
\usepackage{multirow}%
\usepackage{amsmath,amssymb,amsfonts}%
\usepackage{amsthm}
\usepackage{mathrsfs}%
\usepackage[title]{appendix}%
\usepackage{xcolor}%
\usepackage{textcomp}%
\usepackage{manyfoot}%
\usepackage{booktabs}%
\usepackage{algorithm}%
\usepackage{algorithmicx}%
\usepackage{algpseudocode}%
\usepackage{listings}%
\usepackage{todonotes}
\usepackage{mathtools}

\newtheorem{lemma}{Lemma}
\newtheorem{theorem}[lemma]{Theorem}
\newtheorem{corollary}[lemma]{Corollary}

\begin{document}

\title[Disassembling Assembly Theory]{\vspace{-3.3cm}Assembly Theory is an approximation to algorithmic complexity based on LZ compression that does not explain selection or evolution}
\author[1,2,3]{Felipe S. Abrah\~{a}o}
\author[1]{Santiago Hern\'andez-Orozco}
\author[4,5]{Narsis A. Kiani}
\author[6]{Jesper Tegn\'er}
\author[ 1,5,7,8]{Hector Zenil\footnote{Corresponding author. Email: hector.zenil@cs.ox.ac.uk}}


\affil[1]{\small Oxford Immune Algorithmics, Reading, U.K\small}
\affil[2]{\small Center for Logic, Epistemology and the History of Science, University of Campinas, Brazil\small}
\affil[3]{\small DEXL, National Laboratory for Scientific Computing, Brazil\small}
\affil[4]{\small Department of Oncology-Pathology, Center for Molecular Medicine, Karolinska Institutet, Sweden\small}
\affil[5]{\small Algorithmic Dynamics Lab, Center for Molecular Medicine, Karolinska Institutet, Sweden\small}
\affil[6]{\small Living Systems Lab, KAUST, Thuwal, Kingdom of Saudi Arabia\small}
\affil[7]{\small The Alan Turing Institute, British Library, London, U.K\small}
\affil[8]{\small School of Biomedical Engineering and Imaging Sciences, King's College London, U.K\small 
\vspace{-1cm}}

\abstract{
We prove the full equivalence between Assembly Theory (AT) and Shannon Entropy via a method based upon the principles of statistical compression renamed `assembly index' that belongs to the LZ family of popular compression algorithms (ZIP, GZIP, JPEG). Such popular algorithms have been shown to empirically reproduce the results of AT, results that have also been reported before in successful applications to separating organic from non-organic molecules and in the context of the study of selection and evolution.
We show that the assembly index value is equivalent to the size of a minimal context-free grammar. The statistical compressibility of such a method is bounded by Shannon Entropy and other equivalent traditional LZ compression schemes, such as LZ77, LZ78, or LZW. 
In addition, we demonstrate that AT, and the algorithms supporting its pathway complexity, assembly index, and assembly number, define compression schemes and methods that are subsumed into the theory of algorithmic (Kolmogorov-Solomonoff-Chaitin) complexity.
Due to AT's current lack of logical consistency in defining causality for non-stochastic processes and the lack of empirical evidence that it outperforms other complexity measures found in the literature capable of explaining the same phenomena, we conclude that the assembly index and the assembly number do not lead to an explanation or quantification of biases in generative (physical or biological) processes, including those brought about by (abiotic or Darwinian) selection and evolution, that could not have been arrived at using Shannon Entropy or that have not been reported before using classical information theory or algorithmic complexity.
   
}




\maketitle

\section{Introduction}\label{sec1}

Assembly Theory (AT) has recently garnered significant attention in some academic circles and the scientific media. Responding to how open-ended forms can emerge from matter without a blueprint, AT purports to explain and quantify selection and evolution. 
The central claim of AT advanced in~\cite{cronin,croninnature} is that objects with a high assembly index ``are very unlikely to form abiotically''. 
This has been contested in~\cite{hazen}, whose results ``demonstrate that abiotic chemical processes have the potential to form crystal structures of great complexity'', exceeding the assembly index threshold (or MA, when applied to quantify the assembly index on molecules) proposed by AT's authors.
The existence of such abiotic objects would render AT's methods prone to false positives, corroborating the predictions in~\cite{salient}.
The assembly index (or MA) was shown to perform equally well, or worse in some cases, relative to popular compression algorithms (including those of a statistical nature)~\cite{salient}, some of which have been applied before in \cite{zenilchem,genait,zenildna}.

The following sections show the connections of AT to well-established theories that AT's authors claim are different from and unrelated to~\cite{croninentropy}.
We show that the claim advanced by its authors that AT unifies life and biology with physics~\cite{croninnature,prat} relies on a circular argument and on the use of a popular compression algorithm, with no empirical or logical support to establish deeper connections to selection and evolution than those already known (hierarchical modularity), already made (in connection to complexity)~\cite{genait,iscience,Abrahao2021bEmergenceAIDPTRSA}, or previously investigated using information- and graph-theoretic approaches to chemical and molecular complexity~\cite{Ivanciuc2013ChemicalGraphsMolecular,Mowshowitz2012,Bottcher2018MoleculesLife,zenilchem}.

The non-linear (or ``tree''-like) structure of the minimum rooted assembly (sub)spaces is one of the ways to define a (compression) scheme that encodes the assembling process itself back into a linear sequence of codewords. 
We demonstrate that the generative/assembling process in an assembly space is strictly equivalent to a \emph{Context-Free Grammar} (CFG), while the \emph{assembly index} value is equivalent to the size of such a compressing grammar.
This also implies that the assembly index can only be an approximation to the number of factors~\cite{10.1007/978-3-540-27836-8_5} in a traditional LZ scheme, such as LZ78 or LZW, which are statistical compression methods whose compression rates converge to that of Shannon Entropy.

In addition to the assembly index calculation method being a CFG-based compression method bounded by traditional LZ schemes, we demonstrate that the assembling process that results in the construction of an object is a \emph{LZ scheme}.
The compressibility achieved by this LZ scheme depends on the length of the shortest assembling paths and the ``simplicity'' of the minimal assembly spaces.
This notion of simplicity encompasses both how the assembly space structure differs from a single-thread (or linear) space and how many more distinct assembling paths can lead to the same object.
Such a LZ scheme reduces the complexity (which the assembly index aims to quantify) of the assembling process of an object to the compressibility and computational resource efficiency of a compression method.

The theory and methods of AT describing a compression algorithm are, therefore, subsumed into Algorithmic Information Theory (AIT), which in turn has been applied in the same areas that AT has covered, from organic versus non-organic compound classification to bio- and technosignature detection and selection and evolution~\cite{zenilchem,zenilfirst,genait,zenilld}. This makes the assembly index and assembly number proposed by AT effectively approximations to algorithmic (Kolmogorov-Solomonoff-Chaitin) complexity~\cite{Chaitin2004,Calude2002main,liandvitanyi,Downey2010}.

As proposed in \cite{croninnature}, the \emph{assembly number} is intended to measure the amount of selection and evolution necessary to produce the ensemble of (assembled) objects. 
That is, the assembly number aims to quantify the presence of constraints or biases in the underlying generative processes (e.g., those parts of the environment in which the objects were assembled) of the ensemble, processes that set the conditions for the appearance of the assembled objects.
A higher assembly number---not to be conflated with the assembly index---would mean that more ``selective forces'' were in play as, e.g., environmental constraints or biases, in order to allow or generate a higher concentration of high-assembly-index elements. Otherwise, these high-assembly-index objects would not occur as often in the ensemble.

Consonant with the constraints and biases that the assembly number aims to quantify, though in fact, as we demonstrate, it constitutes a \emph{compression method} subsumed into AIT, ensembles with higher assembly numbers are more compressible and would therefore diverge more markedly from those ensembles that are more statistically random, having fewer constraints and biases.
If one assumes that the assembly number quantifies the presence of constraints or biases in the underlying generative processes of the ensembles, then a more compressible ensemble implies that more constraints or biases played a role in generating more high-assembly-index objects more frequently than would have been the case in an environment with fewer constraints and biases (i.e., a more random or incompressible environment), thus increasing the frequency of occurrence of high-assembly-index (i.e., less compressible) objects in this environment.

Conversely, under the same assumption, the presence of more biotic processes in an ensemble implies a higher assembly number, which in turn would imply that the ensemble is more compressible.
This occurs, for example, in scenarios where there is a stronger presence of top-down (or downward) causation \cite{Abrahao2021bEmergenceAIDPTRSA} behind the possibilities or paths that lead to the construction of the objects, while a less compressible ensemble would indicate a weaker presence (or absence) of top-down causation.

We therefore conclude that AT cannot offer a different or better explanation of selection, evolution, or top-down causation than the connections already established~\cite{genait,Abrahao2021bEmergenceAIDPTRSA}, consistent with our previous position that a single scalar is unlikely to classify life or quantify selection or evolution independent of the environment and the perturbations it imposes on the objects and on the assembly process.

\section{Disassembling Assembly Theory: Main Concerns}

\subsection{The assembly index is a compression algorithm of the LZ family}
Despite the authors' assertion in the Assembly Theory (AT) paper that this theory and the assembly index are unrelated to algorithmic complexity, it is evident that AT is fundamentally encompassed within the realm of algorithmic complexity \cite{salient}. 
The assembly index, as proposed, seeks to gauge the complexity of an object based on the number of steps in its shortest copy-counting assembly pathway \cite{croninnature} via a procedure equivalent to LZ compression, which in turn is a computable approximation to algorithmic complexity, denoted by $\mathbf{K}$. 
See Sup. Inf.~\ref{supmat}.

At its core, the assembly index shares three key elements of entropy defined by Shannon himself~\cite{shannon}, and shares the three main elements of the LZ family of compression algorithms~\cite{lz,ziv1978compression} that make LZ converge to Shannon Entropy at the limit: the identification of repeated blocks, the usage of a dictionary containing repetitions and substitutions, and the ability to losslessly reconstruct the original object using its minimal LZ description. 
As demonstrated in the Sup. Inf.~\ref{supmat}, the parity in number of steps for compression and decompression using the assembly index effectively reduces the definition of AT to the length of the \emph{shortest assembly pathway}, which in turn is a loose upper bound of a resource-bounded approximation to algorithmic 
complexity $\mathbf{K}$.

In contrast, more robust approximations to $ \mathbf{K} $, capable of capturing blocks and other causal content within an object, have been proposed for purposes ranging from exploring cause-and-effect chains to quantifying object memory and characterising process content~\cite{zenilbook,bdm}. These more advanced measures have found application in the same contexts and domains explored by AT, encompassing tasks like distinguishing organic from non-organic molecules~\cite{zenilchem}, investigating potential connections to selection and evolution~\cite{genait,iscience}, the detection of bio- and technosignatures \cite{zenil2012ImageCharacClassif,zenilet}, and explorations into causality \cite{nmi,aidbook}. Importantly, when applied to the data employed as evidence by AT, the more sophisticated measures consistently outperform the assembly index (see Sup. Inf.~\ref{supmat} and \cite{salient}). Thus, it becomes evident that AT and its assembly index represent a considerably constrained version of compression algorithms, and a loose upper bound of $ \mathbf{K} $. 

This limitation is attributable to the authors' exclusive consideration of computer programs adhering to the form of `Template Program A', as follows:\\

\noindent\begin{minipage}{0.30\linewidth}
\textbf{Template Program A:}\\
\end{minipage} \begin{minipage}{.70\linewidth}`while end-of-object, do $N$ times print(repetitions) $+$ print (all remaining objects not found in repetitions)'
\end{minipage}\\

Any approximation to $ \mathbf{K} $ that accounts for identical repetitions, including all known lossless statistical compression algorithms, can achieve equivalent or superior results, as demonstrated in the Supplementary Information and \cite{salient}. This alignment with `Template Program A' effectively highlights the association of AT with well-established principles of compression and coding theory, thereby refuting the initial claim of its authors to present a unique methodology. Moreover, the authors' suggestion that their index may be generalised as a universally applicable algorithm for any object (including text)~\cite{fridman} further underscores the 
disconnect between AT-- and its authors' drive to reinvent traditional algorithms such as text compression based on Shannon Entropy--and the current state of the art in the field of statistical and non-statistical compression beyond LZW~\cite{aidbook}.

It is worth noting that although there may be minor variations in the implementation details of the assembly index and the assembly number, of which the authors themselves have proposed significantly different versions, we demonstrate that they are all qualitatively and quantitatively equivalent to the LZ compression algorithms introduced in the 1970s, such as LZ77/LZ78. See Sup. Inf.~\ref{supmat}.

In essence, the assembly index is fundamentally underpinned by the LZ encoding of the objects it measures. This reveals that the assembly index, and consequently AT, aligns more closely with the principles of traditional information theory (Shannon Entropy), and statistical data compression than the authors are willing to acknowledge~\cite{croninentropy}.
In fact, both AT and the statistical compression methods that underpin it are subsumed into the theory of algorithmic complexity.

The AT authors also assert that their index's ability to differentiate between organic and non-organic compounds validates its natural applicability. However, when compared with other statistical indexes, including various compression algorithms, these alternative methods often result in similar or superior performances~\cite{salient}. This undermines the claims made in favour of AT.

\subsection{A compression algorithm is a mechanical procedure that corresponds to a physical process}

The authors of AT argue that traditional data compression algorithms are an overly abstract process unsuited for modelling the construction (or assembly) of objects~\cite{fridman}. This view overlooks the practical and mechanistic nature of compression algorithms, particularly those in the LZ family. Since their introduction in 1977, LZ algorithms have been effectively used in detection, identification, clustering, and classification across various fields, including biology, chemistry, and medicine ~\cite{licompression,zenilchem,dauwels2011slowing}, and to approximate algorithmic complexity $\mathbf{K}$~\cite{liandvitanyi,zenilreview}.

The argument presented by the proponents of AT regarding the uncomputability of algorithmic complexity is deeply misguided. While it is true that $\mathbf{K}$ is semi-computable, computable algorithms like LZ77/LZ78/LZW have been widely used to approximate it. These algorithms have been applied in biology and chemistry, challenging the assertion that AT represents a unique or superior approach ~\cite{zenilchem,salient}.  Furthermore, in a 1976 article \cite{lempel1976complexity}, Lempel and Ziv defined an early version of their algorithm directly as a computable method to approximate algorithmic complexity, defining what is known as LZ complexity, which speaks to the initial motivation behind their now seminal compression algorithm.

The emphasis on the purported requirement of a Turing machine in the context of algorithmic complexity is a misdirected concern. 'Turing machine' is synonymous with `algorithm' and an 'algorithm' is a synonym for a 'Turing machine'. Any rule or method that is algorithmic in nature, including all the computable methods in AT, are algorithms, and are therefore technically Turing machines or programs that can run on a (universal) Turing machine. These facts are only a technicality unrelated to any putative advantage of AT over other approaches to algorithmic complexity.

On the contrary, one of the distinctive features of algorithmic information theory (AIT), algorithmic complexity being one of its indexes, is that its importance and pervasiveness in mathematics, theoretical computer science, and complexity science are in fact owed to the invariance of its results with respect to the model of computation, whether abstract or physically implemented.
More specifically, such statements from the authors of AT~\cite{fridman} indicate a fundamental misunderstanding of equivalence classes closed under reductions (in this case, computability classes), which are basic and pivotal concepts in computer science, complexity theory, and mathematics in general. It is evidence in support of this point that none of the formal proofs in the Sup. Inf. section~\ref{supmat} make any mention of any Turing machine.
In the same manner, all Turing machines referenced in \cite{salient} can be replaced by any sufficiently expressive programming language running on an arbitrary computer.

In practice, computable approximations like LZ77/LZ78---which AT mimics without attribution---do not require a Turing machine, and have been widely and successfully used in clustering and categorisation ~\cite{licompression,liandvitanyi}, including in the successful separation of chemical compounds into organic and non-organic categories~\cite{zenilchem}, and in the reconstruction of evolutionary phylogenetic trees~\cite{licompression}.
To say that a physical (computable/recursive) process like those affected by AT's methods is not a compression algorithm because it does not resemble or correspond to the functioning of a Turing machine is as naively wrong and misplaced as saying a program written in Python cannot model the movement of a pendulum because nature does not run Python. A Turing machine is an abstraction and a synonym of an `algorithm.' Everything is an algorithm in Assembly Theory, and therefore it is governed by the same principles of computer science and information theory.

This type of argument also reflects a misunderstanding of the purposes of Turing machines and of the foundations of computer science. As explained in a quote often attributed to Edsger W. Dijkstra,
\begin{quote}
``Computer science is no more about computers than astronomy is about telescopes."
\end{quote}
At its inception, Turing machines were defined within computer science
as an abstraction of the concept of an algorithm, years before what we now know as computers were built. 

In a published paper~\cite{croninentropy},
the authors offered a proof of computability of their assembly index to distance themselves from Shannon Entropy and algorithmic (Kolmogorov) complexity.  Their algorithm, categorised as \textbf{Template Program A}, is trivially computable and requires no proof of computability, but all other resource-bounded approximations to $\mathbf{K}$ are also computable, including LZW that has been used for 60 years for 
similar purposes~\cite{liandvitanyi}.
\subsection{Conflating object assembly process and directionality of causation}

The authors assume and present the sequential nature of their algorithm as an advantage~\cite{fridman}. 
Assembly Theory (AT) claims to be able to extract causal knowledge by measuring the degree of causality in the form of non-unique chains of cause and effect across the assembly pathways~\cite{fridman}. 
To tackle the problem of top-down (or downward) causation~\cite{Abrahao2021bEmergenceAIDPTRSA}, its assembly number should
be able to ``detect the emergence of new levels of organisation and their causal influence on lower-level phenomena in the world you are observing''~\cite{Jaeger2024ATspringer}.

The authors' central assumption is that each step of the basic \textbf{Template Program A} constitutes a cause-and-effect chain corresponding to how an object may have been physically assembled from an assembly path, which may not be unique. 
This is no different from, and is indeed a restricted version of, a (deterministic or non-deterministic) pushdown automaton capable of instantiating a grammar compressor (which the authors rename an `assembly pathway') and works exactly like a LZ algorithm. 
See Sup. Inf.~\ref{supmat}.

The assumption of sequential assembly from instantaneous object snapshots is not based on any physical (biological, chemical, or other) evidence. In contrast, there is overwhelming evidence that this is not the case. For example, in the case of a genome sequence, all the regions of the genome sequence are exposed to selective forces simultaneously. Transcription factors, or genes that regulate other genes, do not get assembled or interact only sequentially. In chemistry, reactions do not happen on one side of a molecule first and then propagate to the other but happen in parallel. Reactions in time follow and are represented sequentially. However, throughout AT's arguments, causal directionality in time is conflated with how an object may have been assembled from its instantaneous configuration.

This assumption of sequential assembly of an object based on causal direction is manifestly incorrect in regards to how an object subject to selection assembles.
In contrast, only by finding the generative mechanism of the object---such as the underlying set of (computable) mechanisms (of which the assembly index only takes into account one instance, the program that only counts copies)---and thus explaining the object in a non-trivial fashion, can one reproduce both the causal direction from the sequence of connected steps (see Fig.~\ref{figureTreediagrams}B) and how the object itself may have been assembled (which is not and cannot be in a sequential fashion). 

AT and its index cannot characterise this, for example, because the class of problems or functions recognised by pushdown automata (or generated by context-free grammars) is a proper subset of the class of recursive problems.

The authors of AT have mistakenly claimed that algorithmic complexity and Turing machines are not, or cannot be, related to causality. This is incorrect \cite{nmi,zenil2023algorithmic}. They were introduced as causal artifacts for studying mechanistic means of performing logical operations.

We have demonstrated that the assembly index requires a finite automaton to instantiate a context-free grammar (see Sup. Inf.~\ref{supmat}), a version of a specific-purpose machine, just as any other resource-bounded approximation to $\mathbf{K}$ would need to be instantiated (e.g. calculated, even by hand). In this case, the size of their formula or the size of the implementation of their computer program is the size of the special-purpose finite automaton, which is common and of fixed size for their calculations (which allows them to be discounted from the final length, just as it is from $\mathbf{K}$ or resource-bounded approximations).

That the calculation of the assembly index from the minimum rooted assembly space can be reduced to a grammar compression scheme demonstrates not only that the assembly process' complexity is dependent on such contingencies as which assembly spaces or which distinct paths were taken in past historical stages of building the object,  but also that this recursivity and dependence on past trajectories are subsumed into the concept of compressibility.
However, unlike algorithmic approaches, assembly theory does not incorporate any of the elements of causality such as perturbation and counterfactual analysis, as has been done in the context of algorithmic complexity~\cite{aidbook}, a generalisation that when compared to AT reflects the concealed simplistic nature of AT.


\subsection{Lack of control experiments and absence of supporting empirical data beyond current domain knowledge}

A critical examination of the AT methodology reveals significant shortcomings. Firstly, proponents of AT failed to conduct basic control experiments, a foundational aspect of introducing a new scientific metric. Benchmarking against established indices, particularly in coding and compression algorithms, is crucial to validating any new metric in the domain. Previous work on AT
has never included meaningful experimental comparisons of the assembly index with other existing measures on false grounds that their measure is completely different~\cite{croninentropy} (Figs. 1 and 2). Yet, we have shown that other algorithms, such as RLE, Huffman coding (the first dictionary-based universal compression algorithm), and other compression algorithms based on dictionary-based methods and Shannon Entropy produced equivalent or superior results compared to the results published by the authors of AT~\cite{salient}.

AT also introduces a cutoff value in its index, purported to offer a unique perspective on molecular complexity, distinct from those based on algorithmic complexity. This value indicates when an object is more likely to be organic, alive, or a product of a living system.  However, various indices and compression algorithms tested have yielded equivalent cutoff values (see~\cite{salient}). 
In fact, such an Assembly Index threshold replicates previous results on molecular separation that employ algorithmic complexity.
Thus, this overlap in results challenges the notion that AT provides a unique tool for distinguishing between organic and non-organic matter. For instance, similar cutoff values have been derived in other studies, such as those mentioned in~\cite{zenilchem}, effectively separating organic from non-organic molecular compounds. 
While the authors of AT ignored decades of research in chemical complexity based on graph theory and the principles of Entropy~\cite{Mowshowitz2012,Ivanciuc2013ChemicalGraphsMolecular,vonKorff2019MolecularComplexity,Bottcher2018MoleculesLife}, all of which are equivalent to AT, the depth of the purported experimental validation of AT has been notably limited, especially when contrasted with more comprehensive studies. For instance, a more exhaustive and systematic approach, with proper control experiments involving over 15,000 chemical compounds~\cite{zenilchem} and employing algorithms from the LZ family (or Entropy) and others of greater sophistication (BDM), demonstrated the ability to separate organic from non-organic molecular compounds.

Turing's motivation was to explore what could be automated in the causal mechanisation of operations in Fig.~\ref{figureTreediagrams}B, for example, stepwise by hand using paper and pencil. The shortest among all the computer programs of this type is an upper bound approximation to $\mathbf{K}$. In other words, $\mathbf{K}$ cannot be longer than the length of this diagram. The concept of pathway complexity and the algorithm for copy number instantiated by the assembly index represent a restricted version of a Turing machine or finite automaton. Moreover, a Turing machine is a synonym of an algorithm and vice versa. It is a mistake to think of a Turing machine as a physical object or an object with particular properties or features (such as a head or a tape). All the algorithms in assembly theory are Turing machines, and vice versa.

\begin{figure}[ht!]
\centerline{\includegraphics[scale=0.23]{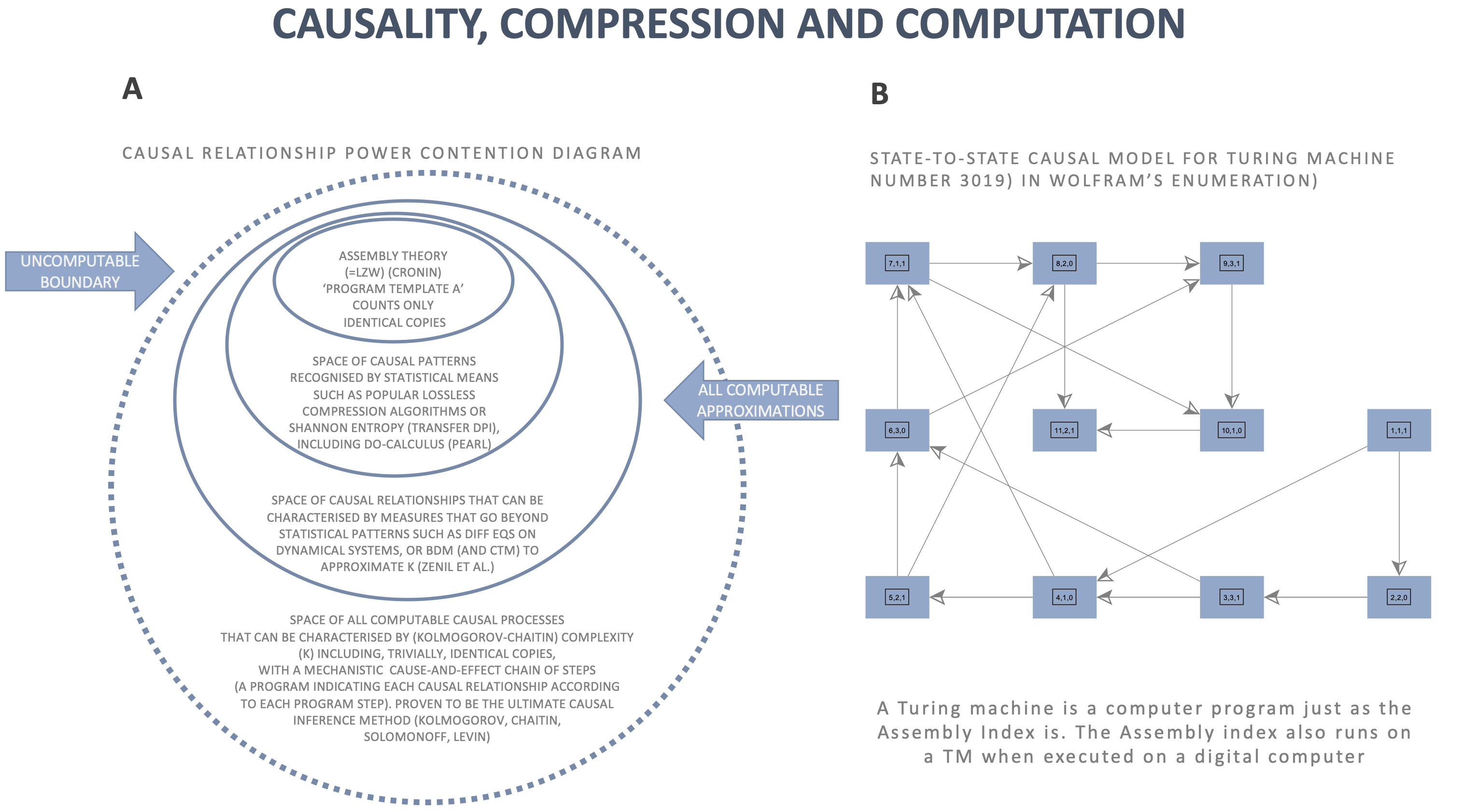}}
\caption{\label{figureTreediagrams} A: The authors of AT have suggested that $ \mathbf{K} $ would be proven to be contained in AT~\cite{fridman}. This Venn diagram shows how AT is connected to and subsumed within algorithmic complexity or ($ \mathbf{K} $) and within the group of statistical compression as proven in this paper (see Sup. Inf.~\ref{supmat}).  B: Causal transition graph of a Turing machine with number 3019 (in Wolfram's enumeration scheme~\cite{Wolfram2002}) with an empty initial condition found by using a computable method (e.g. CTM, \cite{zenil2011}) to explain how the block-patterned string 111000111000 was assembled step-by-step based on the principles of $ \mathbf{K} $ describing the state, memory, and output of the process as a fully causal mechanistic explanation. A Turing machine is simply a procedural algorithm and any algorithm can be represented by a Turing machine. By definition, this is a mechanistic process, and as physical as anything else, not an `abstract' or `unrealisable' process.}
\end{figure}


\subsection{A circular argument cannot unify physics and biology}

Central to Assembly Theory (AT) and the public claims made through the authors' university press releases~\cite{prat} is the connection that they make to selection and evolution, maintaining that AT unifies physics and biology and explains and quantifies selection and evolution, both Darwinian and abiotic~\cite{croninnature}. To demonstrate this connection, assuming selectivity in the combination of linear strings $(P)$, the authors compare two schemes for combining linear strings $Q$, one random versus a non-random selection. 
This experiment yields observations of differences between the two $(R)$, and the authors conclude that selectivity $(S)$ exists in molecular assembly, and therefore that AT can explain it.

Observing that a random selection differs from a non-random selection of strings, the authors use this as evidence for ``the presence of selectivity in the combination process between the polymers existing in the assembly pool''. 

Yet, assuming selectivity in the combination of strings and proceeding to say that a selection algorithm is expected to differ from a random one makes for a circular argument. Furthermore, the conclusion $(S)$ reaffirms the initial assumption $(P)$, resulting in a circular argument. 
In a formal propositional chain, it can be represented as $P \implies Q \implies R \implies S \implies P$, where the conclusion merely restates the initial assumption, lacking any validation or verification, empirical or logical, beyond a self-evident tautology.

Whether this sequential reasoning is relevant to how molecules are actually assembled is also unclear, given the overwhelming evidence that objects are not constructed sequentially and that object complexity in living systems is clearly not driven by identical copies only~\cite{antiredundancy}. 

Nevertheless, the question of random \textit{versus} non-random selection and evolution in the context of approximations to algorithmic complexity, including copy counting, was experimentally tested, empirically supported, and reported before in~\cite{genait} following proper basic principles such as a literature search and review, control experiments (comparison to other measures), and validation against existing knowledge in genetics and cell biology.

These circular arguments, lacking a foundation in chemistry, biology, or empirical evidence, lead the authors to propose the concept of ``assembly time''. The suggestion is a time-scale separation between molecule production (assembly time) and discovery time, aiming to unify physics and biology. Time, however, has always been fundamental in evolutionary theory, and claims about an object's history are the foundations of evolutionary theory. AT revisits existing complexity science concepts described in the timeline provided in Fig.~\ref{timeline} but without acknowledgment or attribution.

\begin{figure}[ht!]
\centerline{\includegraphics[scale=0.24]{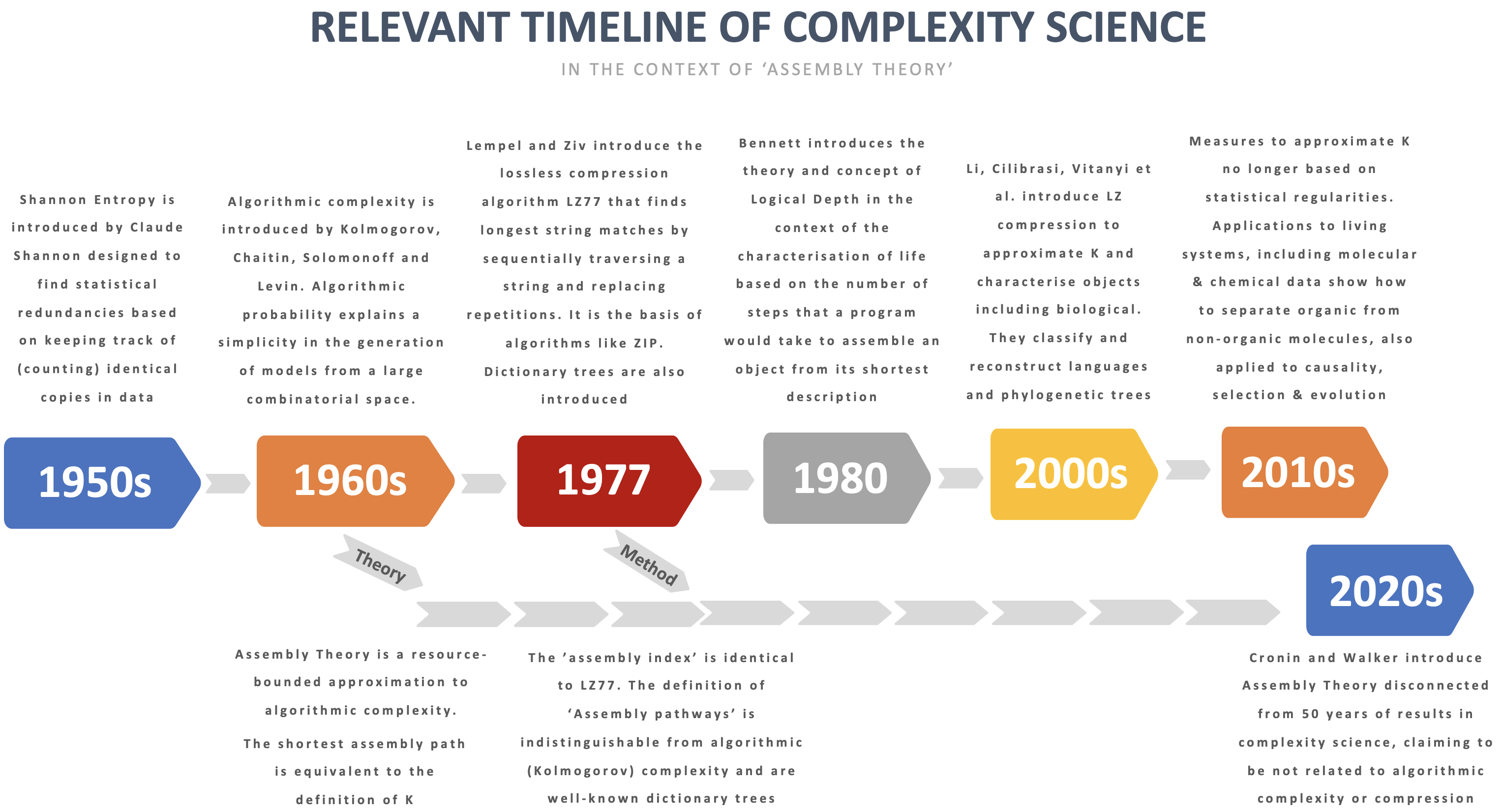}}
\caption{\label{timeline} A timeline of results in complexity science relevant to the claims and results of AT. AT renames several concepts, e.g. dictionary trees as `assembly pathways'; relies heavily on algorithmic probability in its reduction of combinatorial space arguments, without attribution; and, as demonstrated, the assembly index is a LZ compression scheme (proofs provided in the Sup. Inf.~\ref{supmat})
}
\end{figure}

\section{Conclusions}

The most popular examples used by the authors of Assembly Theory (AT) in their papers as illustrations of the way their algorithms work, such as ABRACADABRA and BANANA, are traditionally used to teach LZ77, LZ78, or LZW compression in Computer Science courses at university level.
Dictionary trees have been used for pedagogic purposes in computer science for decades.
This would suggest some similarity between AT and compression algorithms, particularly those of a statistical nature based on recursion to previous states, copy counting, and the re-usage of patterns and repetitions.

In this paper, we have shown that the calculation of the assembly index value from the minimum rooted assembly (sub)space is a compression scheme belonging to the LZ family of compression algorithms, rather than merely being similar. Diving into the details of AT's algorithms and methods, we have revealed that they are equivalent to LZW's, and therefore to (Shannon) Entropy-based methods.

Our results demonstrate that an object with a low assembly index has high LZ compressibility (i.e., it is more compressible according to the LZ scheme), and therefore would necessarily display low entropy when generated by i.i.d. stochastic processes (or low entropy rate in ergodic stationary processes in general). In the opposite direction, an object with a high assembly index will have low LZ compressibility, and therefore high entropy.
This means that the methods based on AT, Shannon Entropy, and LZ compression algorithms are indistinguishable from each other with regard to the quantification of complexity of the assembling process, and the claim that these other approaches are incapable of dealing with AT's type of data (capturing `structure' or anything that could supposedly not be characterised by Shannon Entropy or LZW) is inaccurate.

This paper also demonstrates that AT is subsumed into the theory and methods of Algorithmic Information Theory (AIT) as illustrated in Fig.~\ref{figureTreediagrams}A and mathematically proven in the Sup. Inf. section~\ref{supmat}.

In addition, both theoretical concepts and empirical results in~\cite{cronin,marshall_murray_cronin_2017,croninentropy,croninnature} have been reported previously in earlier work in relation to
chemical processes in~\cite{zenilchem}, and to biology in~\cite{zenilchem,zenilfirst,genait,zenilld} but with comparison to other measures, including methods that go beyond traditional statistical compression and are connected to the concept of causal discovery~\cite{iscience,nmi}.

Algorithmic complexity has been shown to be profoundly connected to causality~\cite{bookaid,scholarpedia,nmi,Abrahao2021bEmergenceAIDPTRSA}, and not only in strings. It has also been applied to images~\cite{zenilld}, networks~\cite{iscience,zenilnet,zenilnet2,zenilnet3}, vectors, matrices, and tensors in multiple dimensions~\cite{zenil2d,Abrahao2021AIDistortionsEntropy,Abrahao2020cAIDistortionsCN}; and to chemical structures~\cite{zenilchem}.
In contrast to the statistical compression schemes on which AT is based, and whose assembly index method is a particular case of a compression scheme, exploring other resource-bounded computable approximations to algorithmic complexity~\cite{bdm,nmi} that consider aspects other than traditional statistical patterns such as identical repetitions,
may have more discriminatory power.
This is because they may tell apart cases with a high assembly index value 
and an expectedly low copy number (or frequency of occurrence in the ensemble) from those with low LZ compressibility and high entropy, 
seeming to indicate that an object is more causally disconnected, independent, or statistically random, while actually, it is strongly causally dependent due to non-trivial rewriting rules not captured by statistical means. This is because AT is to statistical correlation, as measures of algorithmic complexity are to causation. Rarely the two are the same (correlation is not causation) and when they are, they are fully captured by other traditional statistical measures such as Shannon Entropy without the introduction of a methodological different framework that adds no more to the correlation problem beyond Shannon Entropy.

If AT has some elements that appear to have discriminatory power~\cite{marshall_murray_cronin_2017,cronin}, it is because of its connection to Shannon Entropy.

One positive aspect of AT is that it offers a graph-like representational approach (albeit more convoluted than necessary) to approximating LZ compression in the specific context of molecular complexity, potentially making it more accessible to a broader and less technically-minded community. However, it is inaccurate and bad practice to ignore, omit or imply the absence of previous work (e.g. \cite{zenilchem,genait,zenildna,licompression}), belittle and devalue the theories and ideas which AT's indexes and measures are entirely dependent upon and connected to (Shannon Entropy, compression, and Kolmogorov complexity), and make disproportionate public claims related to the contributions of AT (with the media and the authors themselves going so far as to call it a theory that `unifies biology and physics' and a `theory of everything').

We agree with the authors on the importance of understanding the origin and mechanisms of the emergence of an open-ended generation of novelty~\cite{santiagoopenendedpaper,Abrahao2021bEmergenceAIDPTRSA}. However, these efforts by the authors are undermined by hyperbolic claims. 

We argue that the results and claims made by AT, in fact, highlight a working hypothesis that information and computation underpin the concepts necessary to explain the processes and building blocks of physical and living systems, concepts which the methods and frameworks of AT have been heavily inspired by~\cite{universaldist,marvinminsky}.

\section{Acknowledgements}

Felipe S. Abrah\~{a}o acknowledges the partial support of the São Paulo Research Foundation (FAPESP), Grant $2023$/$05593$-$1$.

\bibliography{sn-bibliography.bib}

\newpage

\begin{appendices}

\section{Supplementary information}
\label{supmat}

\subsection{Mathematical framework}\label{sectionDefinitions}


As in \cite{Marshall2019}, let $ \left( \Gamma , \phi \right) $ denote either an \emph{assembly space} or an \emph{assembly subspace}.
In the interests of simplifying notation, one can refer to $ \left( \Gamma , \phi \right) $ simply as $ \Gamma $ where the edge-labellings given by $ \phi $ are not relevant.
From \cite{Marshall2019}, we have it that $ c_\Gamma\left( x \right) $ denotes the \emph{assembly index} of the object $ x $ in the assembly space $ \Gamma $.

Let $ \mathcal{ S } = \left( \mathbf{ \Gamma } , \mathbf{ \Phi } , \mathcal{ F }   \right) $ be an \emph{infinite assembly space}, where every assembly space $ \Gamma \in \mathbf{ \Gamma } $ is finite, $ \mathbf{ \Phi } $ is the set of the corresponding edge-labelling maps $ \phi_\Gamma $ of each $ \Gamma $, and $ \mathcal{ F } = \left( f_1 , \dots , f_n , \dots \right) $ is the infinite sequence of embeddings (in which each embedding is also an \emph{assembly map} as in \cite{Marshall2019}) that ends up generating $ \mathcal{ S } $.
That is, each $ \, f_i \colon \left\{ \Gamma_i \right\} \subseteq \mathbf{ \Gamma } \to \left\{ \Gamma_{ i + 1 } \right\} \subseteq \mathbf{ \Gamma } \, $ is a particular type of assembly map that embeds a single assembly subspace into a larger assembly subspace so that the resulting sequence of nested assembly subspaces defines a total order $ \preceq_{ \mathcal{ S } } $, where 
\[ 
\left( \Gamma_i , \phi_{ \Gamma_i } \right) \preceq_{ \mathcal{ S } } \left( \Gamma_{ i + 1 } , \phi_{ \Gamma_{ i + 1 } } \right)  \text{ \textit{iff} }   f_i\left( \Gamma_i  \right) = \Gamma_{ i + 1 } 
\text{ .}
\]
In other words, $ \mathcal{ S } $ is the nested family of all possible finite assembly spaces from the same basis set
(i.e., the root vertex that represents the set of all basic building blocks) such that $ \mathcal{ S } $ is infinite and enumerable by a program.\footnote{Note that an alternative proof of the results presented in this paper can be achieved just with the (more general) assumption that the assembly space may be finite but only needs to be sufficiently large in comparison to the deceiving program in Lemma~\ref{thmLargeMA}.}

The source vertex $ x $ of the edge $ e \in E\left( \Gamma \right) $ is given by the function $ s_\Gamma\left( e \right) = x $.
The target vertex $ y $ of the edge $ e \in E\left( \Gamma \right) $ is given by the function $ t_\Gamma\left( e \right) = y $.
Let $ \gamma = z \dots y $ denote an arbitrary path from $ z \in B_\mathcal{ S } $ to some terminal $ y \in V\left( \mathcal{ S } \right) $ in $ \mathcal{ S } $, where $ B_\mathcal{ S } $ is the \emph{basis set} (i.e., the finite set of basic building blocks) of $ \mathcal{ S } $ and $ V\left( \mathcal{ S } \right) $ is the set of vertices of $ \mathcal{ S } $.
Let $ \gamma_x $ denote a rooted path from some $ z \in B_\mathcal{ S } $ to the object $ x \in V\left( \mathcal{ S } \right) $.
All of these apply analogously to $ \Gamma $ instead of $ \mathcal{ S } $.

As in \cite{Marshall2019}, let $ { \Gamma^* }_{ x } $ denote a \emph{minimum rooted assembly subspace} of $ \Gamma $ from which the \emph{assembly index} 
\[ \begin{array}{lccc}
	c_\Gamma \colon & \Gamma \subset \mathcal{ S }  & \to & \mathbb{N} \\
	& x \in V\left( \Gamma \right)  & \mapsto & c_\Gamma\left( x \right)
\end{array} \]
calculates the augmented cardinality $ c_\Gamma\left( x \right) $, i.e., the number of vertices in this minimum rooted assembly subspace, except for those in the basis set $ B_\Gamma $.
The \emph{longest} rooted paths $ \gamma_x^{ \max } $ of $ { \Gamma^* }_{ x } $ end in the object/vertex $ x \in V\left( \Gamma \right) $, i.e., $ x $ is the terminal vertex of the path $ \gamma_x^{ \max } $ or, equivalently, $ x $ is the target vertex of the latest oriented edge in $ \gamma_x^{ \max } $.
The \emph{shortest} rooted paths $ \gamma_x^{ \min } $ of $ { \Gamma^* }_{ x } $ end in the object/vertex $ x \in V\left( \Gamma \right) $, i.e., $ x $ is the terminal vertex of the path $ \gamma_x^{ \min } $ or, equivalently, $ x $ is the target vertex of the latest oriented edge in $ \gamma_x^{ \min } $.

Let
$ \; A\# \colon \mathbf{X} \to \mathbb{N} \; $ be a function from the set $ \mathbf{X} $ of possible ensembles (composed of objects) into the set of natural numbers. 
As in \cite{croninnature}, this function defines the \emph{assembly number} $ A\#\left( X \right) $ of an ensemble $ X $ such that
\begin{equation}\label{equationAssemblynumber}
	\begin{array}{lccc}
		A\# \colon & \mathbf{X}  & \to & \mathbb{N} \\
		& X & \mapsto & A\#\left( X \right) =
			A\left( N_T , N , \left( n_1 , \dots , n_N \right) , \left( a_1 , \dots , a_N \right) \right)
	\end{array} 
	\text{ ,}
\end{equation}
where:
\begin{itemize}
	\item $ N_T $ is the total number of objects in the ensemble $ X $;
	
	\item $ N $ is the total number of unique objects in the ensemble $ N_T $;
	
	\item each $ n_i $ corresponds to the copy number of the (unique) object $ x_i $ with $ 1 \leq i \leq N $;
	
	\item $ a_i = c_\Gamma\left( x_i \right) $, where $ \Gamma $ is the assembly space that contains the object $ x_i $;
	
	\item $ A\left( \cdots \right) $ is the \emph{assembly of the ensemble} $ X $ calculated by \[ A\left( N_T , N , \left( n_1 , \dots , n_N \right) , \left( a_1 , \dots , a_N \right) \right) 
	=
	\sum\limits_{ i = 1 }^{ N } e^{ a_i } \left( \frac{ n_i - 1 }{ N_T } \right) \text{ .} \] 
\end{itemize}

According to the assumptions and claims in \cite{cronin,Marshall2019,croninnature} (as also discussed in \cite{salient}), notice that the values of $ N_T $, $ N $, $ n_i $, $ a_i $ are retrievable by an effective procedure from a given (finite) ensemble $ X $, computational methods which are proposed by Assembly Theory \cite{cronin,Marshall2019,croninnature}.
Thus, function $ A\# $ is a recursive function (i.e., a function that an algorithm can calculate) given a finite ensemble $ X $, a property which the authors of \cite{cronin,Marshall2019,croninnature}falsely claim to be one of the advantages of Assembly Theory over algorithmic complexity, because this template was computable already to begin with, so there was no room for uncomputability claims.


Let $ \left< V\left( \Gamma  \right) \right> $ be an encoding of the set $ V\left( \Gamma  \right) $ of vertices of the assembly space (or subspace) $ \Gamma $.
Notice that the encoding formalism and the ordering of the vertices belonging to $ \Gamma $ are arbitrary.
The only condition is that once the encoding method is (arbitrarily) chosen, it is fixed and remains the same for any possible assembly space (or subspace) $ \Gamma $.
Analogously, let $ \left< \gamma \right> $ denote an arbitrarily fixed encoding of a sequence/path $ \gamma $ composed of contiguous oriented edges together with the starting vertex/object in the basis set $ B_\Gamma \subseteq B_{ \mathcal{ S } } $, the sequence which defines the path $ \gamma $ in $ \Gamma $.

We have it that the (prefix) \emph{algorithmic complexity}, denoted by
$ \mathbf{K}\left( x \right) $, is the length of the shortest prefix-free (or self-delimiting) program that outputs the encoded object $ x $ when this program is run using a (universal) programming language.

In accordance with the claims made in \cite{cronin,marshall_murray_cronin_2017,Marshall2019,croninnature,croninentropy}, note that the computability and feasibility of Assembly Theory's methods directly imply the existence of algorithms that can enumerate the assembly spaces.
Given any set of (physical, chemical, and/or biological) rules for assembling molecules that Assembly Theory arbitrarily chooses to be based on, one has it that for any particular pathway assembly resulting in a molecule, there is a corresponding algorithm for which one can apply the methods described in \cite{cronin,marshall_murray_cronin_2017,Marshall2019} to calculate the pathway probability, and so on.
This existence is trivially guaranteed to hold as long as the calculation of the assembly index is recursive from a given assembly space, which is already an assumed condition in \cite{cronin,marshall_murray_cronin_2017,Marshall2019};
moreover, any upper bound for the computational resources necessary to compute the assembly index straightforwardly implies the existence of an upper bound for the computational resources necessary for a program to run.
For every permitted generative process that can assemble objects and thereby build another object, there is an algorithm that performs this process.
Conversely, for every object, one has it that there is a corresponding generative process allowed by Assembly Theory, a process which is effected by an algorithm.
Therefore, according to the assumptions in \cite{cronin,marshall_murray_cronin_2017,Marshall2019,croninnature,croninentropy}, there is a formal theory $ \mathbf{F} $ that contains Assembly Theory and all the procedures effected by algorithms, including the chosen method for calculating the assembly index of an object, the program that decides whether or not the criteria for building the assembly spaces are met, and so on.

\subsection{Assembly theory is an approximation to algorithmic complexity}\label{sectionoAToverestimation}

One can demonstrate in Lemma~\ref{thmLargeMA} and Corollary~\ref{thmInfmanyLargeMA} that the assembly index $ c_\Gamma\left( y \right) $ (or the assembly number $ A\#\left( X \right) $ ) is an approximation to algorithmic complexity $ \mathbf{K}\left( y \right) $ (or, respectively, $ \mathbf{K}\left( X \right) $) that in general \emph{overestimates} its value as much as desired.

As in \cite{salient}, the main route to achieving the following theoretical results is to construct a randomly generated program that receives a formal theory (which contains all the procedures and statistical criteria in Assembly Theory) as input.
Then, it searches for an object $ y $ in an assembly space with a sufficiently high assembly index $ c_\Gamma\left( y \right) $ so as to overcome any arbitrarily chosen error margin $ k $ between $ c_\Gamma\left( y \right) $ and the algorithmic complexity $ \mathbf{K}\left( y \right) $.

\begin{lemma}\label{thmLargeMA}
	Let $ \mathcal{ S } $ be enumerable by an algorithm.
	Let $ \mathbf{F} $ be an arbitrary formal theory that contains Assembly Theory, including all the decidable procedures of
	the chosen method for calculating the assembly index of an object for a nested subspace of $ \mathcal{ S } $,
	and the program that decides whether or not the criteria for building the assembly spaces are met.
	Let $ k \in \mathbb{N} $ be an arbitrarily large natural number.
	Then, there are 
	$ \Gamma \subset \mathcal{ S } $ and $ y \in V\left( \Gamma \right) $ such that 
	\begin{equation}\label{eqLargeMA1}
		\mathbf{K}\left( y \right) + k 
		\leq
		c_\Gamma\left( y \right)
		\text{ ,}
	\end{equation}
	where the function $ c_\Gamma\left( y \right) $ gives the assembly index of the object $ y $ in the assembly space $ \Gamma $ (or $ \mathcal{ S } $).

	\begin{proof}\label{proofthmLargeMA}
		Let $ p $ be a bit string that represents an algorithm that receives $ \mathbf{F} $ and $ k $ as inputs.
		Then, it calculates $ \left| p \right| + \left| \mathbf{F} \right| + \mathbf{O}\left( \log_2\left( k \right) \right) + k $ and enumerates $ \mathcal{ S } $ while calculating $ c_\Gamma\left( x \right) $ of the object (or vertex) $ x \in V\left( \Gamma \right) \subset V\left( \mathcal{ S } \right) $ at each step of this enumeration.
		Finally, the algorithm returns the first object $ y \in  V\left( \mathcal{ S } \right) $ for which 
		\begin{equation}\label{eqProofLargeMA}
			\left| p \right| + \left| \mathbf{F} \right| + \mathbf{O}\left( \log_2\left( k \right) \right) + k + \mathbf{O}(1) \leq c_\Gamma\left( y \right)
		\end{equation}
		holds.
		In order to demonstrate that $ p $ always halts, just note that $ \mathcal{ S } $ is infinite computably enumerable.
		Also, for any value of $ c_{ \Gamma' }\left( z \right) $ for some $ z \in V\left( \Gamma' \right) \subset V\left( \mathcal{ S } \right)  $, there is only a finite number of minimum rooted assembly subspaces (starting on any object in $ B_\mathcal{ S } $ and ending on $ z $) whose augmented cardinality is $ c_{ \Gamma' }\left( z \right) $,
		where $ B_\mathcal{ S } $ is the basis (i.e., the finite set of basic building blocks \cite{cronin}) of $ \mathcal{ S } $.
		This implies that there is an infinite number of distinct values of $ c_{ \Gamma' }\left( z \right) $.
		Now, let $ \mathrm{ p }_y = \left< k , \mathbf{F} , p \right> $. 
		Finally, from Equation~\eqref{eqProofLargeMA} and basic properties in AIT, we have it that
		\begin{equation}\label{eqProofLargeMA2}
			\begin{aligned}
				\mathbf{K}\left( y \right) + k
				\leq
				\left| \mathrm{ p }_y \right| + k + \mathbf{O}(1)
				\leq 
				\left| p \right| + \left| \mathbf{F} \right| + \mathbf{O}\left( \log_2\left( k \right) \right) + k + \mathbf{O}(1) \leq c_\Gamma\left( y \right)
			\end{aligned}
		\end{equation}
		holds for some sufficiently large $ k $.
	\end{proof}
\end{lemma}

From this proof of Lemma~\ref{thmLargeMA}, one can directly conclude that:

\begin{corollary}\label{thmInfmanyLargeMA}
	There are infinitely many $ \Gamma \subset \mathcal{ S } $ for which Equation~\eqref{eqLargeMA1} in Lemma~\ref{thmLargeMA} holds.
	
	\begin{proof}
	
		From the proof of Lemma~\ref{thmLargeMA} we know that there are infinitely many distinct values of $ c_{ \Gamma' }\left( z \right) $.
		Let $ p_2 $ be a slight variation of program $ p $ in the proof of Lemma~\ref{thmLargeMA} such that, in addition to $ \mathbf{F} $ and $ k $, it also receives a natural number $ i \in \mathbb{N} $ as input.
		Then, it enumerates $ \mathcal{ S } $ while calculating $ c_\Gamma\left( x \right) $ of the object (or vertex) $ x \in V\left( \Gamma \right) \subset V\left( \mathcal{ S } \right) $ at each step of this enumeration.
		In the last step, program $ p_2 $ returns the $ i $-th first object $ y \in V( \mathcal{ S } ) $ for which \begin{equation}\label{eqProofCorollaryInmanyLargeMA}
					\left| p \right| + \left| \mathbf{F} \right| + \mathbf{O}\left( \log_2\left( k \right) \right) + k + \mathbf{O}\left( \log_2\left( i \right) \right) + i + \mathbf{O}(1) \leq c_\Gamma\left( y \right)
	\end{equation}
	holds.
	(Notice that Equation~\eqref{eqProofCorollaryInmanyLargeMA} differs from Equation~\eqref{eqProofLargeMA}).
	Then, for each $ i > 0 $, one will obtain a distinct object $ y $ such that
	\begin{equation}\label{eqProofLargeMA3}
		\begin{aligned}
			\mathbf{K}\left( y \right) + k
			\leq 
			\left| p \right| + \left| \mathbf{F} \right| + \mathbf{O}\left( \log_2\left( k \right) \right) + k + \mathbf{O}\left( \log_2\left( i \right) \right) + i + \mathbf{O}(1) \leq c_\Gamma\left( y \right)
		\end{aligned}
	\end{equation}	
	holds.
	\end{proof}
	
\end{corollary}

\subsubsection{Assembly number and algorithmic complexity}\label{sectionDistortionsforAN}

Now we can extend Lemma~\ref{thmLargeMA} to the assembly number $  A\# $, a value which is defined upon an ensemble $ X $.
The main idea of Lemma~\ref{thmLargeAssemblynumber} is that the class of all possible ensembles of objects that assembly spaces can possibly cover is also enumerable by an algorithm, based on the fact that assembly spaces are enumerable by an algorithm.
In this manner, as done in Lemma~\ref{thmLargeMA}, one can make the error margins (in comparison to algorithmic complexity) as large as one wishes.

\begin{lemma}\label{thmLargeAssemblynumber}
	Let $ \mathcal{ S } $ be enumerable by an algorithm.
	Let $ \mathbf{F} $ be an arbitrary formal theory that contains Assembly Theory, including: 
	all the decidable procedures of
	the chosen method for calculating the assembly index $ c_\Gamma $ of an object for a nested subspace of $ \mathcal{ S } $ and the assembly number $ A\# $ of a given ensemble $ X $;
	the program that decides whether or not the criteria for building the assembly spaces are met;
	and the program that decides whether or not the criteria for building ensembles of objects are met. 
	Let $ k \in \mathbb{N} $ be an arbitrarily large natural number.
	Then, there is $ y \in \mathbf{X} $ such that 
	\begin{equation}\label{eqLargeAN1}
		\mathbf{K}\left( X \right) + k 
		\leq
		A\#\left( X \right)
		\text{ ,}
	\end{equation}
	where the function $ \; A\# \colon \mathbf{X} \to \mathbb{N} \; $ gives the assembly number of the ensemble $ X \in \mathbf{X} $.

	\begin{proof}\label{proofthmLargeAssemlynumber}
		Remember the definition of assembly number $ A\# $ such that $ A\#\left( X \right) = \sum\limits_{ i = 1 }^{ N } e^{ a_i } \left( \frac{ n_i - 1 }{ N_T } \right) $, where the parameters $ N_T $, $ N $, $ n_i $, $ a_i $ are computably retrievable from a given ensemble $ X $.
		Since $ \mathcal{S} $ is infinite computably enumerable and the theory $ F $ can always decide whether or not an arbitrary set $ Y $ is an ensemble from which $ A\# $ can be calculated, one has it that the set $ \mathbf{X} $ (of all possible finite ensembles) is infinite computably enumerable. 
		From here on, the proof follows analogously to the proof of Lemma~\ref{thmLargeMA}.
	\end{proof}
\end{lemma}

In the same manner as in Corollary~\ref{thmInfmanyLargeMA} with respect to Lemma~\ref{thmLargeMA}, one can obtain the following corollary:

\begin{corollary}\label{thmInfmanyLargeAN}
	There are infinitely many $ X \in \mathbf{X} $ for which Equation~\eqref{eqLargeAN1} in Lemma~\ref{thmLargeAssemblynumber} holds.
	
	\begin{proof}
		Analogous to the proof of Corollary~\ref{thmInfmanyLargeMA}.
	\end{proof}
	
\end{corollary}

The reader is invited to note that the latter Lemma~\ref{thmLargeAssemblynumber} and Corollary~\ref{thmInfmanyLargeAN} can also be generalised for any recursive function $ f $ defined on the basis of 
the assembly index and on an arbitrary observable (for which Assembly Theory, in particular, would include a formal criterion to decide whether or not function $ f $ can be calculated from this observable).
For example, it applies to the arbitrary function $ A''\#\left( X \right) $ from Section~\ref{sectionEncodingfromAN}.


\subsection{Assembly theory is an encoding/compression scheme}\label{sectionEncodingMA}

\subsubsection{The assembly index defines an encoding scheme}\label{sectionEncodingfromtheIndex}

The following results show that the minimum rooted assembly (sub)space, from which assembly index $ c_\Gamma\left( y \right) $ is calculated, constitutes a parameter that can be employed to encode the minimal assembling process of an object $ y $.
More formally, Theorem~\ref{thmEncodingfromMA} demonstrates that any arbitrarily chosen encoding of the set of vertices (in any arbitrary ordering) of a minimum rooted assembly space is sufficient to encode the highest-assembly-index object assembled across this space.
In other words, the sheer number of objects in a minimum rooted assembly subspace carries enough information to allow one to always pick the object with the highest assembly index value,
even though the ordering of the objects in this collection may be completely arbitrary in the first place.
This demonstrates that AT's assembly index calculation method is one of the possible methods for encoding objects.

\begin{theorem}\label{thmEncodingfromMA}
	Let $ \mathcal{ S } $ be enumerable by an algorithm.
	Let $ \mathbf{F} $ be an arbitrary formal theory that contains Assembly Theory, including all the decidable procedures of
	the chosen method for calculating the assembly index of an object for a nested subspace of $ \mathcal{ S } $,
	and the program that decides whether or not the criteria for building the assembly spaces are met.
	Let $ \Gamma \subset \mathcal{ S } $ with $ y \in V\left( \Gamma \right) $ be arbitrary.
	Then, one has it that $ \left< V\left( { \Gamma^* }_{ y }  \right) \right> $ encodes the object $ y $ such that
	\begin{equation}\label{eq1EncodingfromMA}
			\mathbf{K}\left( y \right)
			\leq
			\mathbf{K}\left(  \left< V\left( { \Gamma^* }_{ y }  \right) \right> \right)
			+ \mathbf{O}(1)
		\end{equation}
	holds, where:
	$ c_\Gamma\left( y \right) $ gives the assembly index of the object $ y $ in the assembly space $ \Gamma $ (or $ \mathcal{ S } $);
	$ { \Gamma^* }_{ y } $ is the minimum rooted assembly subspace of $ \Gamma $ from which the assembly index $ c_\Gamma\left( y \right) $ is calculated;
	and $ \left< V \right> $ is an arbitrarily fixed encoding of a set $ V $ of vertices.
	
	\begin{proof}
		Remember that $ { \Gamma^* }_{ y } $ denotes the minimum rooted assembly subspace of $ \Gamma $ from which the assembly index calculates the augmented cardinality that defines the value $ c_\Gamma\left( y \right) $.
		Notice that since $ { \Gamma^* }_{ y } $ is minimal, then there is a simple algorithm that can always calculate the integer value of $ c_\Gamma\left( y \right) $ given $ \mathbf{K}\left(  \left< V\left( { \Gamma^* }_{ y }  \right) \right> \right) $ as input.
		Remember that, from our definitions and conditions in Section~\ref{sectionDefinitions}, $ \left< V\left( { \Gamma^* }_{ y }  \right) \right> $ denotes an encoding of the set of vertices of $ { \Gamma^* }_{ y } $, an encoding method which is arbitrary, previously chosen, and does not depend on the choice of the assembly (sub)space or object.
		Now, one can exploit the computability of Assembly Theory and construct a program $ p_3 $ that receives $ \mathbf{F} $ 
		and $ \left< V\left( { \Gamma^* }_{ y }  \right) \right> $ as inputs.
		By using $ \mathbf{F} $ to decide whether or not an arbitrary $ \Gamma' $ satisfies the criteria of being an assembly subspace, program $ p_3 $ enumerates the finite set $ \mathbf{\Gamma}\left( V\left( { \Gamma^* }_{ y }  \right) \right) $ of all possible assembly subspaces whose set of vertices is exactly $ V\left( { \Gamma^* }_{ y }  \right) $.
		Next, using $ \mathbf{F} $ to decide whether or not an arbitrary $ \Gamma' $ satisfies the criteria of being a \emph{minimum} rooted assembly subspace whose longest path ends in the same last vertex of $ \Gamma' $,
		program $ p_3 $ searches in $ \mathbf{\Gamma}\left( V\left( { \Gamma^* }_{ y }  \right) \right) $ for the first assembly subspace $ \Gamma'' $ that is minimal. 
		Finally, it returns the object $ y' \in V\left( \Gamma'' \right) $ at which the longest rooted paths in $ \Gamma'' $ end (i.e., $ y' $ is the terminal vertex of these longest rooted paths).
		In order to achieve the desired proof, one has yet to show that the object built at the end of the longest paths in the minimum rooted assembly subspaces is unique and that $ y' = y $.  
		To this end, one can employ the minimality of $ { \Gamma^* }_{ y } $.
		First, suppose there is a minimum $ \Gamma_{2} \in \mathbf{\Gamma}\left( V\left( { \Gamma^* }_{ y }  \right) \right) $ that ends at the object $ z $ such that $ \Gamma_{2} \neq \Gamma'' $ and $ z \neq y' $.
		Then, by construction of the program $ p_3 $, one would have it that $ y' $ belongs to at least one of the rooted paths in $ \Gamma_{2} $.
		Hence, there would be a rooted path that leads to $ y' $ that is shorter than the one that leads to $ z $ because by hypothesis one has it that the path that leads to $ z $ is the longest one; and if there were another terminal vertex $ w $ in $ \Gamma_{2} $ whose rooted path has the same length as that of $ z $, then $ \Gamma_{2} $ would not be a minimum subspace.
		As a consequence, there would be a rooted assembly subspace $ { \Gamma_{2}^* }_{ y' } \subseteq  \Gamma_{2} $ whose augmented cardinality is smaller than $ c_{ \Gamma'' }\left( y' \right) $.
		However, by construction, we have it that $ c_{ \Gamma'' }\left( y' \right) $ is minimal, which leads to a contradiction.
		Therefore, since by hypothesis we have it that $ { \Gamma^* }_{ y } $ is minimal and $ { \Gamma^* }_{ y } \in  \mathbf{\Gamma}\left( V\left( { \Gamma^* }_{ y }  \right) \right)  $, we conclude that either $ \Gamma_{2} = \Gamma'' = { \Gamma^* }_{ y } $ or $ z = y' = y $, and hence that, in any case, program $ p_3 $ produces $ y $ as output.
		Finally, the proof of Equation~\eqref{eq1EncodingfromMA} follows from basic inequalities in AIT.
	\end{proof}
\end{theorem}

Notice that $ \left< V\left( { \Gamma^* }_{ y }  \right) \right> $ is only built to encode the membership relation for the set, and does not need to carry any a priori information about its internal structure that results in the assembly space from which it only takes the set of vertices.
Only by applying a third-party decoding algorithm, such as program $ p_3 $ in the proof of Theorem~\ref{thmEncodingfromMA}, does the final object $ y $ become uniquely retrievable. 

One can demonstrate that from encoding one\footnote{ There might be more than one of these paths. Corollary~\ref{thmEncodingfromlongestpaths} holds in any case, i.e., it holds for any arbitrarily chosen $ \gamma_y^{ \max } $ in $ { \Gamma^* }_{ y } $.} of the longest\footnote{ Notice that a directly analogous version of Corollary~\ref{thmEncodingfromlongestpaths} holds for $ \gamma_y^{ \min } $ instead of $ \gamma_y^{ \max } $. From encoding any of the shortest paths in a minimum rooted assembly subspace, one can always directly retrieve a minimum rooted assembly subspace, and hence the terminal vertex that corresponds to the object to be assembled.} paths in a minimum rooted assembly subspace, one can always retrieve a minimum rooted assembly subspace, and hence the terminal vertex that corresponds to the object to be assembled:

\begin{corollary}\label{thmEncodingfromlongestpaths}
	Let $ \mathcal{ S } $ be enumerable by an algorithm.
	Let $ \mathbf{F} $ be an arbitrary formal theory that contains Assembly Theory, including all the decidable procedures of
	the chosen method for calculating the assembly index of an object for a nested subspace of $ \mathcal{ S } $,
	and the program that decides whether or not the criteria for building the assembly spaces are met.
	Let $ \Gamma \subset \mathcal{ S } $ with $ y \in V\left( \Gamma \right) $ and $ \gamma_y^{ \max } $ be arbitrary, where $ \gamma_y^{ \max } $ is one of the longest rooted paths in $ { \Gamma^* }_{ y } $ that has $ y $ as its terminal vertex and $ { \Gamma^* }_{ y } $ is the minimum rooted assembly subspace of $ \Gamma $ from which the assembly index $ c_\Gamma\left( y \right) $ is calculated.
	Then, one has it that $ \left< \gamma_y^{ \max } \right> $ encodes the object $ y $ such that
	\begin{equation}\label{eq1thmEncodingfromlongestpaths}
		\mathbf{K}\left( y \right)
		\leq
		\mathbf{K}\left(  \left< \gamma_y^{ \max } \right> \right)
		+ \mathbf{O}(1)
	\end{equation}
	holds, where $ \left< \gamma \right> $ is any arbitrarily fixed encoding of a sequence $ \gamma $ composed of contiguous oriented edges together with the starting vertex/object in the basis set $ B_\Gamma \subseteq B_{ \mathcal{ S } } $, the sequence which defines the path $ \gamma $.
	
	\begin{proof}
		Let $ p_3 $ be the program used in the proof of Theorem~\ref{thmEncodingfromMA}.
		Let $ p_4 $ be a program that receives $ \mathbf{F} $, program $ p_3 $, and $ \left< \gamma_y^{ \max } \right> $ as inputs.
		Then, by using $ \mathbf{F} $ to decide whether or not an arbitrary $ \Gamma' $ satisfies the criteria of being a minimum rooted assembly subspace, it
        enumerates all possible minimum rooted assembly subspaces that have $ \gamma_y^{ \max } $ as one of its longest rooted paths.
		In the next step, program $ p_4 $ picks the first of these possible minimum rooted assembly subspaces and encodes it in the form $ \left< V\left( { \Gamma^* }_{ y }  \right) \right> $.
		Finally, program $ p_4 $ returns the output of program $ p_3 $ with $ \mathbf{F} $ and $ \left< V\left( { \Gamma^* }_{ y }  \right) \right> $ as inputs.
		The proof of Equation~\eqref{eq1thmEncodingfromlongestpaths} follows from Equation~\eqref{eq1EncodingfromMA} and basic inequalities in AIT.
	\end{proof}
\end{corollary}

A sequence of edges in each of these longest (or shortest) paths constitutes a way to parse the object (e.g., a string), i.e., to decompose it into contiguous subsequences.
Each distinct assembling path in $ { \Gamma^* }_{ y } $ corresponds to a distinct way to parse the object $ y $.
This idea will become clearer in the next Sections~\ref{sectionLZGrammar} and~\ref{sectionLZcompressionassemblyspaces}.
It also constitutes the compression scheme underlying LZ algorithms.

\subsubsection {The assembly index defines a grammar of compression that is bounded by LZ compression and can at most approximate the Shannon Entropy rate}
\label{sectionLZGrammar}

From here on, we investigate how well the assembly index can also produce an encoding whose size should be as small as possible in comparison to the size of the original object.
In other words, we investigate how the assembly index performs as a \emph{compression method}.
Indeed, we show in the following sections that Assembly Theory is not only an encoding but also a \emph{compression scheme}, in particular one that decomposes (i.e., parses) the object into contiguous building blocks, as is done by other methods like LZ compression algorithms and compressing grammars.

As presented in Section~\ref{sectionDefinitions}, recall that the \emph{assembly index} $  c_\Gamma\left( x \right) $ of an object $x$ is defined as the size of the smallest acyclic quiver (directed multigraph) with edge-labellings from  
its own set of vertices\footnote{ And also by adding the property described in Item~$6$ in \cite[Definition 11]{Marshall2019} of assembly spaces. }, a quiver which is called the \emph{assembly (sub)space}, denoted by $ \left( \Gamma , \phi \right) $, or just $ \Gamma $ where the edge-labellings given by $ \phi $ are not relevant.
The ``tree''-like structure of the assembly space maps the generation of an object from a set of basis objects in $ B_\mathcal{ S } $, which is a finite set and remains constant for any object assembled from combinations of the basis objects. 
Each vertex represents an object, and the directed edge relation specifies the \textit{generation hierarchy} (i.e., the label of the oriented edge) of the object  to be combined with the previous object (i.e., the source vertex to which the respective outgoing labelled edges are connected) in order to assemble a new object (i.e., the target vertex to which the respective incoming labelled edges are connected).

Now, given that these minimum rooted assembly subspaces are finite and acyclic, the number of objects that they can represent by an assembling process in a given assembly space is finite. 
Therefore, as demonstrated in the following Theorem~\ref{main}, there will be a \emph{Context-Free Grammar} (CFG) $ G = \left( V_G, \Sigma_G, R, S  \right) $ that generates each of these objects by combining terminal symbols from an alphabet.
These terminal symbols correspond to the basis objects,
where $ V_G $ is the alphabet of $ G $, $ \Sigma_G \subset V_G $ is the set of terminal symbols, and $ R \subseteq  \left( V_G \setminus \Sigma_G \right) \times V_G^* $ is the set of derivation rules, and $ S \subset \left( V_G \setminus \Sigma_G \right) $ is the start symbol.

Grammar-based compression algorithms have been an active area of research~\cite{kieffer_grammar-based_2000,Lehman2002ApproximationAF,10.1007/978-3-540-27836-8_5,ganczorz2018entropy}. The problem of finding the smallest grammar has been studied for several decades and is known to be NP-Complete~\cite{charikar_smallest_2005,Lehman2002ApproximationAF}. Furthermore, the size of these compressing grammars is lower bounded by LZ encoding~\cite{10.1007/978-3-540-27836-8_5} (disregarding time complexity), which means that they converge to LZ in the optimal case for ergodic and stationary stochastic processes, and therefore are bounded by the noiseless compression limit~\cite{Cover2005}.

It is straightforward to show the existence of such a CFG that encodes all the objects generated by an assembly space. 
In the following, we will investigate it further to demonstrate that for any ${ \Gamma^* }_{ x }$ (i.e., a minimal rooted assembly subspace) that generates a string $x$, there exists a grammar $ G $ that \textit{encodes} this ${ \Gamma^* }_{ x }$, encoding not only the final object $ x $ to be assembled, but the whole assembling process rooted in the basis objects. 
In other words, there is at least one CFG that generates $x$ and manifests a correspondence between: the basis objects of both the grammar (terminal symbols) and the assembly subspace; non-basis objects and vertices of ${ \Gamma^* }_{ y }$ with the sets of non-terminal symbols; and edges of ${ \Gamma^* }_{ y }$ with derivation rules of the grammar.\\

\begin{theorem}\label{main}
    Let 
$ { \Gamma^* }_{ x } $
be an arbitrary minimum rooted assembly (sub)space from which the assembly index $c_\Gamma(x)$ is calculated. 
Then, there exists a CFG  $G = \left( V_G, \Sigma_G, R, S  \right)$ that generates $ x $ such that 
$ V\left( { \Gamma^* }_{ x } \right) \setminus  \left\{ x \right\}  = V_G \setminus \Sigma_G $, 
$ \Sigma_G = B_\Gamma \cup \left\{ x \right\} $, 
and 
\begin{equation}
	\left| E\left( { \Gamma^* }_{ x } \right) \right| + \left| B_\Gamma \right| = | R |
\end{equation}
hold, 
where 
$ B_\Gamma $ is the basis set of the assembly space $ \left( \Gamma , \phi \right) $ from which $ x $ is one of its vertices/objects such that $ { \Gamma^* }_{ x } \subseteq \Gamma $. 

\begin{proof}
Let $ { \Gamma^* }_{ x } $ generate $x$. First, we define the set of terminal symbols in $G$ as the set of basis objects in $ { \Gamma^* }_{ x } $ plus the final object $ x $ to be assembled, that is, $ \Sigma_G = B_\Gamma \cup \left\{ x \right\} $. 
Then, let $ V_G \setminus \Sigma_G = V\left( { \Gamma^* }_{ x } \right) \setminus  \left\{ x \right\}  $
be the set of non-terminal symbols.  
Now, we construct the set $R$ of derivation rules of the CFG in the following way: 
we start by adding $S \rightarrow x $ for each basis object $ x $ that belongs to $ V_G \setminus \Sigma_G $ and also for each basis object that belongs to $ \Sigma_G $ to the set of derivation rules\footnote{ Thus, there will be $ 2 \left| B_\Gamma \right| $ rules in the form $S \rightarrow x $.};
then, for each $e_i \in E\left( { \Gamma^* }_{ x } \right)  $ such that $ s_\Gamma(e_i) = x_i$ and $ t_\Gamma(e_i) = y_i $, we construct a corresponding new rule, either $ x_i \rightarrow  x_i z_i $ or $ x_i \rightarrow  z_i x_i  $, where $ z_i \in V\left( { \Gamma^* }_{ x } \right) $ and $ \phi\left( e_i \right) = z_i $, whichever concatenation of combinations builds the object/vertex $ y_i $ (i.e., either $ y_i = x_i z_i $ or $ y_i = z_i x_i $). 
This is possible since each ``assembly step'' (represented by the respective edge) can only combine two objects: one object is the source vertex of the edge; and the other object is the vertex that is the label of this edge.
If $ y_i = x_i z_i \neq x $ (or $ y_i = z_i x_i \neq x $), then we add $ y_i $ to the set of non-terminal symbols $ V_G \setminus \Sigma_G $.
Otherwise, we add $ y_i = x $ to the set of terminal symbols of $ G $.
From this construction of the set $ R $, one will obtain that
\begin{equation}
	\left| E\left( { \Gamma^* }_{ x } \right) \right| + \left| B_\Gamma \right| = | R |
	\text{ .}
\end{equation}
The rest of the proof that $G$ generates $x$ follows by induction over the number of vertices in $ V\left( { \Gamma^* }_{ x } \right) $. 
If $ \left| V\left( { \Gamma^* }_{ x } \right) \right| = 1 $, then $x$ must be a basis object in $ B_\Gamma $. Hence, $ x $ is a terminal symbol in $G$, and thereby it directly follows that $S \rightarrow x$ generates $x$. 
Now, we assume that $ G $ generates any non-terminal symbol $ x_i $ for $ i \leq n $ with a set $ R_n $ of derivation rules, where each vertex $ x_i $ belongs to a rooted path (in $ { \Gamma^* }_{ x } $) whose terminal vertex is $ x $. 
Let $ V\left( { \Gamma^* }_{ x } \right) = \left\{ x_1, \dots, x_n, x_{ n + 1 } , \dots , x \right\} $ be the set of vertices so that  ${ \Gamma^* }_{ x }$ generates $x$. 
By the inductive hypothesis, since there exists an edge $ e $  with $ s_\Gamma\left( e \right) = x_i $ and $ t_\Gamma\left( e \right) = x_j $ for some $ x_i $ and $ x_j $ with $ i,j \leq n $, we have it that there is the rule  $ x_i \rightarrow  x_i z_i $ or $ x_i \rightarrow  z_i x_i  $, where
either $ x_j = x_i z_i $ or $ x_j = z_i x_i $, respectively.
By construction of $ { \Gamma^* }_{ x } $, we have it that there is a vertex $ x'_i \in \left\{ x_1, \dots, x_n \right\} $ that is the source vertex of an edge $ e' \in E\left( { \Gamma^* }_{ x } \right) $ whose label is another $ x'_j \in \left\{ x_1, \dots, x_n \right\} $ and the target vertex of $ e' $ is $ t_\Gamma\left( e' \right) = x_{ n + 1 } $.
Then, we add the rule $ x'_i \rightarrow  x'_i x'_j $ or $ x'_i \rightarrow  x'_j x'_i  $ to $ R_n $, where 
$ x_{ n + 1 } = x'_i x'_j $ or $ x_{ n + 1 } = x'_j x'_i $, respectively. 
Therefore, since by the inductive hypothesis both $ x'_i $ and $ x'_j $ can be generated by the set $ R_n $ of rules, we will have it that $ R_n \cup \left\{ x'_i \rightarrow  x'_i x'_j \lor x'_i \rightarrow  x'_j x'_i \right\} $ generates the object $ x_{ n + 1 } $.
Finally, by the above construction of the set $ R $ of derivation rules, we will have it that if $ x_{ n + 1 } = x $, then $ x_{ n + 1 } $ will be a terminal symbol and $ G $ stops exactly at the object $ x $.

\end{proof}
    
\end{theorem}

The grammar defined in Theorem \ref{main} demonstrates that any assembled object can be encoded into a Context-Free Grammar (CFG) in Chomsky Normal Form (CNF) that captures the same information regarding how an object is constructed from other objects. 
The assembly index is defined as the number of non-basis vertices, while the \emph{size} of the grammar is defined as the number of derivation rules, which is equivalent to the number of nonterminals if there are no useless rules~\cite{10.1007/978-3-540-27836-8_5}. Let's denote this size by \(|G|\). 
In order to encode the object $ x $ assembled according to $ { \Gamma^* }_{ x } $, there is no necessity for having two or more edges linking the same pair of vertices in an assembly space.
Thus, for completeness' sake, we will denote by \(G^*\) the grammar with any useless rules/edges removed. 
We will show in Corollary~\ref{thmMinimumgrammar} that there is a simple linear relationship between the assembly index \(c_{\Gamma}(x)\) and the size of the grammar \(|G^*|\), except for the size of the basis set $ B_\Gamma $ which remains constant for any object assembled upon combinations of basis objects.

\begin{corollary}\label{thmMinimumgrammar}
Let $x$ be a string/object assembled according to $ { \Gamma^* }_{ x } $ and let $G$ be the CFG in Theorem~\ref{main}. Then, 
\begin{itemize}
    \item the assembly index is equivalent to the non-basis symbols in the minimum grammar that generates $ x $, i.e.,
	\begin{equation}
		c_{\Gamma}(x) = |G^*| - \left| B_\Gamma \right|
		\text{ ;}
	\end{equation}

    \item  the number of factors in the LZ encoded form of $x$, denoted by $\text{LZ}(x)$, is upper bounded by the assembly index (except for the basis set $ B_\Gamma $ that is finite and constant), i.e.,
	\begin{equation}
		LZ(x) \leq c_\Gamma(x) + \left| B_\Gamma \right|
		\text{ .}
	\end{equation}
\end{itemize}
\begin{proof}
    The first equation is by construction of $G$ in Theorem~\ref{main}. The second is a direct consequence of the first and \cite[Theorem $1$]{10.1007/978-3-540-27836-8_5}.
\end{proof}
\end{corollary}

Corollary~\ref{thmMinimumgrammar} proves that the assembly index for strings is equivalent to the compression size of a compressing grammar, demonstrating that the notion of the minimum ``amount of assembling steps'' (that the minimum rooted assembly subspace defines) is exactly the same as a (minimal) CFG.
For each of the grammars constructed in Theorem~\ref{main}, the final object $ x $ is unique, and for each distinct assembling path leading to $ x $, there is a (unique) corresponding sequence of rules in $ R $ that also results in the object $ x $.

In the optimal case, such a size can only converge towards a traditional LZ compression scheme~\cite{lempel1976complexity, 1055934, welch1984technique}, such as LZ$77$ and LZ$78$,  and thus its compression rate capabilities are bounded by Shannon Entropy by the (noiseless) Source Coding Theorem \cite{noiselesscoding,Cover2005}. It is important to note that the strict inequality arises when the grammar, thus the pathway assembly, is not optimal at identifying the ``building blocks'' of $x$ given the same alphabet of basis objects.  For compressing grammars in particular, several implementation-specific entropy bounds were found in \cite{ganczorz2018entropy}. 


Now, for higher-dimensional objects, if the construction is recursive, it is possible to find a one-dimensional rule that encodes the object's assembly. For example, for a two-dimensional object, we can construct a grid of fixed size around the object and incorporate the positional information into the grammar's rules while maintaining a linear relationship between the assembly index and the size of the grammar. The rule \(X \rightarrow ApB\) can define how object \(B\) is added to object \(A\) at position \(p\), where \(p\) corresponds to a square in the grid, to generate object \(X\).

\subsubsection{The assembly index is a version of LZ compression}\label{sectionLZcompressionassemblyspaces}

In Section~\ref{sectionEncodingMA}, we have demonstrated that either the set of vertices of the minimum rooted assembly subspace $ { \Gamma^* }_{ y } $ or \emph{at least one} of the longest (or shortest) rooted paths in this $ { \Gamma^* }_{ y } $ is sufficient to encode the object $ y $.
These results employ  general algorithmic methods so that they are independent of the choice of programming language and encoding-decoding schemes.
In Theorem~\ref{main}, we have demonstrated that every minimum rooted assembly (sub)space corresponds to \emph{at least one} context-free grammar (CFG), which reveals a way one can compress an assembly space (e.g., the minimum rooted assembly subspace $ { \Gamma^* }_{ y } $ that generates the object $ y $) into the minimum set of rules of a CFG.
These results together indicate the existence of some LZ-based method to encode both the object and its assembling process.
Indeed, the proof of the following Theorem~\ref{thmLZcompressionfromassemblyspaces} constructs such a LZ compression scheme.

For traditional LZ schemes, the encoded/compressed form of a string comprises a sequence of encoded tuples (or factors) of finite length.
Typically, each one of these tuples contains at least one pointer (e.g., in the case of LZ$78$ and LZW, called indexes) that refers to another tuple or to an entry in a dictionary.
Also, it can contain information about some elements of the basic alphabet (e.g., in the case of LZ$78$ and LZW it contains the codeword for only one symbol to be concatenated at the end of the element the pointer is referring to).
It is trivial to find examples of objects whose parsing according to the LZ$78$ differs from the parsing according to the LZ$77$, or that the one according to the LZMA differs from LZW, and so on.
The LZ family tree of compression schemes includes many distinct schemes, including LZ77, LZ78, LZW, LZSS, and LZMA~\cite{Salomon2010HandbookDataCompression}.
All of these compression methods are defined by schemes that resort to a (static or dynamic) dictionary containing tokens, indexes, pointers, or basic symbols from which the decoder or decompressor can retrieve the original raw data from the received encoded/compressed form according to the respective LZ scheme~\cite{Salomon2010HandbookDataCompression}.
Theorem~\ref{thmLZcompressionfromassemblyspaces} demonstrates that the minimum rooted assembly (sub)space from which the assembly index is calculated is a compression scheme that belongs to this LZ family of compression algorithms.

On the one hand, the pointers in traditional LZ compression schemes, such as LZ77, LZ78 or LZW, refer to encoded tuples that have occurred before as recorded by the associated dictionary.
On the other hand, for Assembly Theory, the pointers refer to other tuples that might not have occurred before in an assembling path, but are part of the minimum rooted assembly subspace $ { \Gamma^* }_{ x } $, that is, part of a ``tree''-like structure that needs to be encoded in the dictionary, which is only built or queried during the encoding/decoding process.

In practice, the major characteristic that distinguishes assembly index compression from other traditional LZ schemes (like LZ77/78 or LZW), is that: 
in the latter case, one would parse the string into factors, that is, contiguous smaller strings (or codewords) whose resulting unidimensional concatenation directly represents the string;
while in the former case, one would parse the object into a non-unidimensional structure where each block is ``concatenated'' according to a directed acyclic multigraph with edge labels drawn from the set of vertices.
For a minimum rooted assembly (sub)space there may be more than one rooted path terminating at the final object to be assembled.
Each of these paths corresponds to a distinct way to parse (i.e., decompose) the object.

The sequence of edges that constitutes an assembling path (from the basis set $ B_\Gamma $ to the terminal object/vertex in a minimum rooted assembly subspace $  { \Gamma^* }_{ y } $) can be converted into a sequence of respective encoded tuples that contains pointers, codewords for other objects, etc.
One can think of these assembling paths being, for example, the sequence defined by either one of the longest rooted paths $ \gamma_x^{ \max } $ or one of the shortest rooted paths $ \gamma_x^{ \min } $ in the minimum rooted assembly subspace $ { \Gamma^* }_{ x } $.
In particular, the shortest path is often cited by AT's authors as the one that translates or captures the idea of complexity that the assembly index is aiming at quantifying \cite{cronin,marshall_murray_cronin_2017,Marshall2019,croninnature,croninentropy}.

Indeed, as demonstrated by the upcoming Theorem~\ref{thmLZcompressionfromassemblyspaces}, the length of the shortest paths is roughly the same as the number of factors (or the number of tuples of pointers) in the encoded form (see Equation~\eqref{eqSassemblespace}) of the LZ scheme that the assembly index calculation method is equivalent to.
Thus, as shown in Theorem~\ref{thmLZcompressionfromassemblyspaces}, the LZ compression scheme defined by the assembly index calculation method establishes a way to compress such a non-linear structure (in particular, a directed acyclic multigraph) and non-deterministic parsings/decompositions back into a linear sequence of codewords.
We demonstrate that the minimality of $ { \Gamma^* }_{ x } $ in fact guarantees that such a scheme (where pointers may not refer to the unidimensional past subsequences) allows unique decoding, despite the non-unidimensional structure of assembly spaces in which more than one path may lead to the same object.\\

\begin{theorem}[LZ encoding from the assembly index calculation method]\label{thmLZcompressionfromassemblyspaces}
	Let $ \mathcal{ S } $ be enumerable by an algorithm.
	Let $ \mathbf{F} $ be an arbitrary formal theory that contains Assembly Theory, including all the decidable procedures of
	the chosen method for calculating the assembly index of an object for a nested subspace of $ \mathcal{ S } $,
	and the program that decides whether or not the criteria for building the assembly spaces are met.
	Let $ \Gamma \subset \mathcal{ S } $ with $ y \in V\left( \Gamma \right) $, $ \gamma_y^{ \min } $, and $ \gamma_y^{ \max } $ be arbitrary.
	Then, one has it that the sequence $ S_{ AT }\left( y \right) $ composed of $ \left| \gamma_y^{ \min } \right| $ encoded tuples in the form $ \left< c , i_t , v_l 
	\right> $ (except for the $ 0 $-th initial tuples which are $ \left< t \right> $, $ \left< \left| \gamma_y^{ \max } \right| \right> $, $ \left< B_\Gamma \right>  $, and $ \left< v_B \right> $, respectively, where $ v_B \in B_\Gamma $ and $ t \in \mathbb{P}\left( \mathcal{ T } \right) $) encodes the assembling process of the object $ y $ such that
		\begin{equation}\label{eq1thmEncodingfromshortestpaths}
			\begin{aligned}
				\mathbf{O}\left( 	\left| \gamma_y^{ \min } \right| \right)
				\leq \\
				\left| S_{ AT }\left( y \right) \right|
				\leq \\
				\mathbf{O}\left( 2^{ 3 \, \left| \gamma_y^{ \max } \right| + 6 } + \left| B_\Gamma \right| \log\left( \left| B_\Gamma \right| \right) \right) 
				+ 
				\mathbf{O}\left( \log\left( \left| \gamma_y^{ \max } \right| \right) \right)
				+ \\
				\mathbf{O}\left( \left( \left| B_\Gamma \right| + 1 \right) \, \log\left( \left| B_\Gamma \right| \right) \right) 
				+
				\left| \gamma_y^{ \min } \right| \Big( 
					\mathbf{O}\left( 3 \, \left| \gamma_y^{ \max } \right| + 6 \right)
					+
					\log\left( \left| \gamma_y^{ \min } \right| \right)
				\Big) 
				\leq \\
				\mathbf{O}\left( 2^{ 3 \, c_\Gamma\left( y \right) + 6 } + \left| B_\Gamma \right| \log\left( \left| B_\Gamma \right| \right) \right) 
					+ 
					\mathbf{O}\left( \log\left( c_\Gamma\left( y \right) \right) \right)
					+ \\
					\mathbf{O}\left( \left( \left| B_\Gamma \right| + 1 \right) \, \log\left( \left| B_\Gamma \right| \right) \right) 
					+
					c_\Gamma\left( y \right) \Big( 
						\mathbf{O}\left( 3 \, c_\Gamma\left( y \right) + 6 \right)
						+
						\log\left( c_\Gamma\left( y \right) \right)
					\Big)
			\end{aligned}
		\end{equation}
		holds, where:
		\begin{itemize}
			\item $ \left| \gamma_y^{ \min } \right| $ is the length of the rooted path $ \gamma_y^{ \min } $, which is one of the shortest rooted paths with $ y $ as terminal vertex that belongs to the minimum rooted assembly subspace $  { \Gamma^* }_{ y } $;
			
			\item $ \left| \gamma_y^{ \max } \right| $ is the length of the rooted path $ \gamma_y^{ \max } $, which is one of the longest rooted paths with $ y $ as terminal vertex that belongs to the minimum rooted assembly subspace $  { \Gamma^* }_{ y } $;
			
			\item $ \left< \left| \gamma_y^{ \max } \right| \right> $ is an encoding of the number $ \left| \gamma_y^{ \max } \right| $;
			
			\item $ \left< B_\Gamma \right>  $ is an encoding of the basis set $ B_\Gamma $;
			
			\item $ \left< v_B \right> $ is an encoding of the vertex $ v \in B_\Gamma $ which is the initial object/vertex of $ \gamma_y^{ \min } $;
			
			\item $ c $ is the counter of the tuple in the sequence $ S_{ AT }\left( y \right) $ (i.e., the natural number that indicates the position of the tuple) with $ 1 \leq c \leq \left| \gamma_y^{ \min } \right| $;
			
			\item $ \left< t \right> $ is an encoding of the $ t $-th member of $ \mathbb{P}\left( \mathcal{ T } \right) $, where $ \mathbb{P}\left( \mathcal{ T } \right) $ is the powerset of possible edges of the minimum directed acyclic multigraph $ \mathcal{ T } $ into which any possible minimum rooted assembly subspace (whose longest rooted paths have length $ \left| \gamma_y^{ \max } \right| $) can be embedded, and all the most distant vertices from the terminal vertex (i.e., $ y $) can only be the vertices in $ B_\Gamma $. 
			
			\item $ i_t $ is the index/pointer that refers to the corresponding edge in $ \mathcal{T} $;
			
			\item $ v_l $ is another index/pointer that refers to the edge $ i'_t $ whose source vertex $ s\left( i'_t \right) = v \in V\left( \Gamma \right) $ is the vertex that is the label of the edge $ i_t $ when $ \Gamma $ is embedded into $ \mathcal{T} $;
			

			\item $ \left| S_{ AT }\left( y \right) \right| $ denotes the size in bits of the encoded sequence $ S_{ AT }\left( y \right) $.
			
			\item $ c_\Gamma\left( y \right) $ gives the assembly index of the object $ y $ in the assembly space $ \Gamma $ (or $ \mathcal{ S } $);
		\end{itemize}

	\begin{proof}[The encoding-decoding scheme for Theorem~\ref{thmLZcompressionfromassemblyspaces}]
		Let $ p_5 $ be a program that receives a sequence $ S_{ AT }\left( y \right) $ in the form
		\begin{equation}\label{eqSassemblespace}
			\left< t \right>
			\left< \left| \gamma_y^{ \max } \right| \right>
			\left< B_\Gamma \right>
			\left< v_B \right>
			\left< 1 , { i_t }_1 , { v_l }_1  \right>
			\left< 2 , { i_t }_2 , { v_l }_2  \right>
			\left< 3 , { i_t }_3 , { v_l }_3  \right>
			\dots
			\left< \left| \gamma_y^{ \min } \right| , { i_t }_{ \left| \gamma_y^{ \min } \right| } , { v_l }_{ \left| \gamma_y^{ \min } \right| }  \right>
		\end{equation}
		as input.
		Then, by employing $ \mathbf{F} $, it constructs the directed acyclic multigraph (DAmG) $ \mathcal{ T } $ of height $ \left| \gamma_y^{ \max } \right| $, and enumerates all the possible minimum rooted assembly subspaces contained in $ \mathcal{ T } $ via the following procedure: it starts from those that are isomorphic to an arbitrarily fixed $ \gamma_y^{ \max } $ with any possible combination of initial vertex/object in $ B_\Gamma $;
		and then in a dovetailing procedure, it indexes first those subspaces that differ (from this $ \gamma_y^{ \max } $) by the least amount of additional edges.
		Then, program $ p_5 $ chooses the $ t $-th first $ \Gamma' $ that appears in this enumeration.
		In the next step, program $ p_5 $ reads the rest of the sequence $  S_{ AT }\left( y \right) $, extracting from the encoded tuples the sequence $ I = \left( { i_t }_1 , \dots , { i_t }_{ \left| \gamma_y^{ \min } \right| } , { v_l }_1 , \dots , { v_l }_{ \left| \gamma_y^{ \min } \right| } \right) $ of indexes of edges that are contained in this $ t $-th first $ \Gamma' $ such that it constitutes a shortest rooted path that starts at vertex $ v_B $.
		Finally, it uses this shortest rooted path to assemble the final object $ y' $ accordingly.
		This describes the \emph{decoding} procedure from $ S_{ AT }\left( y \right) $.
		In order to construct a program $ p'_5 $ that \emph{encodes} a minimum rooted assembly subspace $  { \Gamma^* }_{ y } $ into a $ S_{ AT }\left( y \right) $, one just needs to note that the codewords for 
		$ \left< \left| \gamma_y^{ \max } \right| \right> $, $ \left< B_\Gamma \right> $, 
		and $ \left< y \right> $ are unique and always retrievable by an algorithm given $ { \Gamma^* }_{ y } $ as input.
		From these, program $ p'_5 $ can use the same procedures that $ p_5 $ ran in order to construct $ \mathcal{T} $. 
		Program $ p'_5 $ then employs the same procedures that $ p_5 $ ran in order to enumerate the minimum rooted assembly subspaces embedded into $ \mathcal{T} $ so as to find the index $ t $ that corresponds to $ { \Gamma^* }_{ y } $ in this enumeration.
		Next, it enumerates all the possible shortest rooted paths in this $ t $-th minimum rooted assembly subspace.
		Then, it extracts $ v_B $ and $ I $ from $ \mathcal{T} $ by matching those that correspond to the edges and their respective labels in the first shortest rooted path $ \gamma_y^{ \min } $ found in the previous enumeration that terminates at the object $ y $.
		Finally, the program returns the sequence $ S_{ AT }\left( y \right) $ as output.
	\end{proof}
	
	\begin{proof}[Proof of Theorem~\ref{thmLZcompressionfromassemblyspaces}]
		With the purpose of demonstrating that $ \Gamma' = { \Gamma^* }_{ y } $,
		first note that there is only a finite number of rooted assembly spaces whose most distant vertices (from the vertex/object $ y $) belong to $ B_\Gamma $ (which is also finite and fixed given $ \Gamma $) and are separated by $ \left| \gamma_y^{ \max } \right| $ hops.
		Therefore, since $ \left< \left| \gamma_y^{ \max } \right| \right> $ and $ \left< B_\Gamma \right> $ are 
                 distinct
                 from $ { \Gamma^* }_{ y } $ and the enumeration procedure is the same in both the encoding and the decoding phase, there is a unique index $ t $ that corresponds to $ \Gamma' = { \Gamma^* }_{ y } $.
		In order to achieve the proof of Equation~\eqref{eq1thmEncodingfromshortestpaths}, 
		one can exploit the minimality of $ { \Gamma^* }_{ y } $ and $ \gamma_y^{ \min } $ and the maximality of $ \gamma_y^{ \max } $.
		From any rooted assembly space one of whose longest rooted paths is $ \gamma_y^{ \max } $, in order to construct $ \mathcal{T} $, an algorithm can minimise the number of sufficient and necessary new vertices that need to be added to $ \gamma_y^{ \max } $ using the following procedure:
		assign an index $ j \in \mathbb{N} $ to each of the edges in $ \gamma_y^{ \max } $ starting from the closest and moving to the farthest ones from the initial vertex, which is the object in the basis set $ B_\Gamma $, such that $ 1 \leq j \leq \left| \gamma_y^{ \max } \right| $;
		for the label of the edge $ j = 1 $ in $ \gamma_y^{ \max } $, one only requires at most one extra vertex, and this vertex can only belong to $ B_\Gamma $ because otherwise, since there would need to be an additional edge connecting this new vertex to the target vertex of edge $ j = 1 $, it would contradict the maximality of $ \gamma_y^{ \max } $;
		the same analogously holds for each arbitrary $ j + 1 \leq \left| \gamma_y^{ \max } \right|$, so that labelling this edge $ j +1 $ with a vertex only requires at most one extra new vertex, and this vertex can only be at level $ \leq j $, i.e., it can only belong to the set of vertices that are separated from the terminal vertex $ y $ by at least $ \left| \gamma_y^{ \max } \right| - j + 1 $ hops (otherwise, it would again contradict the maximality of $ \gamma_y^{ \max } $);
		and for each new vertex created at level $ j $ one would only need at most $ 2 $ new vertices at level $ \leq j - 1 $, and so on.
		Thus, by induction, one will have it that
		\begin{equation}\label{eqMaximumverticesinGamma}
			\sum\limits_{ i = 0 }^{ \left| \gamma_y^{ \max } \right| } 2^{ i }
			=
			2^{ \left| \gamma_y^{ \max } \right| + 2 }
		\end{equation}
		is an upper bound for the maximum number of distinct vertices that a minimum rooted assembly subspace (whose $ \gamma_y^{ \max } $ is one of its longest rooted paths) can have.
		Following a similar inductive argument upon the maximality of $ \gamma_y^{ \max } $, one concludes that 
		\begin{equation}\label{eqMaximumedgesinGamma}
			2^{ 3 \, \left| \gamma_y^{ \max } \right| + 6 }
		\end{equation}
		is an upper bound for the maximum number of distinct (multi)edges that a minimum rooted assembly subspace can have.
		Hence, one only needs at most
		\begin{equation}
			\mathbf{O}\left( \log\left( 2^{ 3 \, \left| \gamma_y^{ \max } \right| + 6 } \right) \right)
			=
			\mathbf{O}\left( 3 \, \left| \gamma_y^{ \max } \right| + 6 \right)
		\end{equation}
		bits to encode either each $ i_t $ or each $ v_l $.
		(Notice that there may be tighter upper bounds than those in Equations~\eqref{eqMaximumverticesinGamma} and~\eqref{eqMaximumedgesinGamma}, but for the purposes of Theorem~\ref{thmLZcompressionfromassemblyspaces} they suffice).
		Therefore, we have it that $ { \Gamma^* }_{ y } $ can be embedded into $  \mathcal{T} $, and as a consequence there will be a $ t \in \mathbb{N} $, where
		\begin{equation}\label{eqUpperboundont}
			1
			\leq
			t
			\leq
			\left( 2^{ 2^{ 3 \, \left| \gamma_y^{ \max } \right| + 6 } } \right) \, \left| B_\Gamma \right| !
		\end{equation} 
		holds,
		such that $ t $ corresponds to this exact structure given by $ { \Gamma^* }_{ y } $ when program $ p_5 $ starts to enumerate all possible minimum assembly subspaces embedded into $ \mathcal{T} $.
		Then, one needs at most 
		\begin{equation}
			\mathbf{O}\left( \log\left( \left( 2^{ 2^{ 3 \, \left| \gamma_y^{ \max } \right| + 6 } } \right) \, \left| B_\Gamma \right| ! \right) \right)
			\leq
			\mathbf{O}\left( 2^{ 3 \, \left| \gamma_y^{ \max } \right| + 6 } + \left| B_\Gamma \right| \log\left( \left| B_\Gamma \right| \right) \right)
		\end{equation}
		bits to encode this $ t $.
		By knowing the exact DAmG structure of $ { \Gamma^* }_{ y } $, and to which edge each basis vertex is connected, an algorithm can employ $ \mathbf{F} $ to reconstruct both every object assembled across this DAmG structure and every rooted path that terminates at each of these objects.
		Now, by the definition of assembly space and the assembly index, we have it that $ \left| \gamma_y^{ \min } \right| \leq \left| \gamma_y^{ \max } \right| \leq c_\Gamma\left( y \right) $.
		Therefore, from basic inequalities known from the theory of algorithmic complexity, when applied to the best way to encode each term in the sequence in Equation~\eqref{eqSassemblespace}, one obtains the proof of Equation~\eqref{eq1thmEncodingfromshortestpaths}.
	\end{proof}
\end{theorem}

In relation to the claims made by  
the authors of Assembly Theory\cite{cronin,marshall_murray_cronin_2017,Marshall2019,croninnature,croninentropy}, Theorem~\ref{thmLZcompressionfromassemblyspaces} reveals that Assembly Theory's fundamental intuition to classify the notion of complexity of assembled objects is not only that of a compression scheme, but specifically an LZ scheme.
The non-unidimensional ``tree''-like structure of the minimum rooted assembly subspaces is one of the ways to define a scheme that encodes the assembling process itself back into a linear sequence of codewords.

The encoding-decoding scheme behind Equation~\eqref{eqSassemblespace} translates the notion of the minimum assembling path in an assembly space, playing a role in how much simpler the object can get (i.e., as we have demonstrated, how much more compressible the object is).
Shorter minimum paths lead to more compressed forms of the assembling process of the objects.
The exponential growth in the LZ encoding from the assembly index method is, in turn, consonant with the assembly number \cite{croninnature} of an ensemble (see Section~\ref{sectionEncodingfromAN}).

In the same manner, Theorem~\ref{thmLZcompressionfromassemblyspaces} demonstrates that ``simpler'' minimum rooted assembly subspaces lead to more compressibility.
This notion of simplicity encompasses both how much the assembly space structure differs from a single-thread (or linear) space and how many more distinct assembling paths can lead to the same object.

For example, the reader is invited to note that due to the fact that $ \gamma_y^{ \max } = \gamma_y^{ \min } $ would hold, a linear minimum rooted assembly subspace $ { \Gamma^* }_{ y } $ (such as in \cite[Figure 1(a)]{cronin} or \cite[Figure 8]{Marshall2019}) results in a much tighter upper bound in Equation~\eqref{eq1thmEncodingfromshortestpaths}.
This is because the exponential dependence on the assembly index to encode all possible DAG-like structures would drop: 
that is, $ 1 \leq t \leq \mathbf{O}\left( \left| \gamma_y^{ \max } \right|^2 \, \left| B_\Gamma \right| \right) $ will hold instead of Equation~\eqref{eqUpperboundont}, which in turn will be subsumed in the rightmost term of Equation~\eqref{eq1thmEncodingfromshortestpaths} where one already has a quadratic dependence on $ \left| \gamma_y^{ \max } \right| $ or $ c_\Gamma\left( y \right) $.

Once a less linear (or more DAG-like) assembling structure 
unfolds, as $ \gamma_y^{ \max } $ begins to differ more from $ \gamma_y^{ \min } $, the encoding or decoding algorithms in Theorem~\ref{thmLZcompressionfromassemblyspaces} will have to enumerate more distinct possible minimum rooted assembly spaces, and as a consequence it will increase the upper bound for $ \left| S_{ AT }\left( y \right) \right| $, which also renders the encoding and decoding process slower. 
Finding the minimum rooted assembly subspace for an object is a computationally expensive problem, as demonstrated in Section~\ref{sectionLZGrammar} and acknowledged by AT's authors \cite{cronin,marshall_murray_cronin_2017,Marshall2019,croninnature,croninentropy}.
To tackle this intractability, AT employed some approximation techniques, including a split-branch version of the assembly index.

Thus, Theorem~\ref{thmLZcompressionfromassemblyspaces} not only demonstrates that those minimum rooted assembly subspaces (from which the assembly indexes are calculated) are LZ schemes, but also that the intractability of the calculation of the assembly index (or, equivalently, the intractability of the encoding-decoding procedure) derives from the notion of complexity that the assembly index intends to measure.
Together with the height of the grammars in Theorem~\ref{main}, this also captures the notion of ``depth'' that AT intends to attain but that in fact pertains, as we have demonstrated, to a compression scheme and a generative grammar.
The more complex the DAG-like structure of $ { \Gamma^* }_{ y } $ is, the more distinct DAG-like structures are possible, and therefore the closer the upper bound becomes to the exponential dependence. 
Hence, Equation~\eqref{eq1thmEncodingfromshortestpaths} would render the object (or assembling process) less compressible,
which demonstrates the claims underlying Assembly Theory that the authors in \cite{cronin,marshall_murray_cronin_2017,Marshall2019,croninnature,croninentropy} originally intended to demonstrate with the assembly index but that in fact we demonstrated with a compression scheme.
These results show that the notion of how complex an assembling process needs to be in order to construct an object is indeed equivalent to how much less compressible (by a LZ scheme) this process is.

Further, as in Section~\ref{sectionEncodingfromAN}, one can demonstrate that any measure defined in terms of the assembly index (so that it increases when the assembly index(es) also increases), such as the assembly number~\cite{croninnature}, also defines an encoding and a compression method. The higher the assembly value of an object, the more the object is LZ compressible and the other way around, in direct and equal proportion.

\subsubsection{The assembly number defines an encoding and a compression scheme}\label{sectionEncodingfromAN}

As proposed in \cite{croninnature}, the assembly number $ A\#\left( X \right) = \sum_{ i = 1 }^{ N } e^{ a_i } \left( \frac{ n_i - 1 }{ N_T } \right) $ is intended to measure the amount of selection necessary to produce the ensemble $ X $, i.e., the presence of constraints or biases in the underlying generative processes (e.g., the environment in which the objects were assembled) that set the conditions for the appearance of the objects.
A higher value of $ A\#\left( X \right) $ would mean that more ``selective forces'' were put into play as environmental constraints or biases in order to allow or generate a higher concentration of high-assembly-index elements.
Otherwise, these high-assembly-index objects would not occur as often in the ensemble.
This is because in stochastically random scenarios the probability rapidly decreases as the assembly index increases.

The value of $ A\# $ is calculated from the parameters $ N_T $, $ N $, $ n_i $, $ a_i $, which are retrievable by an algorithm from a given ensemble $ X $ whose number of copies of each unique object is greater than or equal to $ 1 $.
The constant $ N_T $ (i.e., the total number of objects in the ensemble $ X $) allows one to compare the assembly number between distinct ensembles whose sizes differ.

Notice that each term $ e^{ a_i } $ plays the role of weight for the respective ratio/proportion $ \frac{ n_i - 1 }{ N_T } \leq 1 $.
In this manner, one can, for example, normalise again for ensembles whose overall summation of assembly indexes remains the same, hence only differing by the concentration/distribution of the quantity of objects with higher assembly index value.
As stated in \cite{cronin,marshall_murray_cronin_2017,Marshall2019,croninnature}, grasping this distinction in fact is the main goal that Assembly Theory aims at with its methods, so as to measure the presence of constraints or biases that turn otherwise unlikely objects into objects that occur more often.

For any two ensembles $ X $ and $ X' $ with the same total number $ N_T $ of objects and whose overall summation $ \sum_{ i = 1 }^{ N } e^{ a_i } $ of assembly indexes is also the same, the ensemble with the higher concentration of high-assembly-index objects (with respect to the concentration of these objects in the other ensemble) should score higher on the measure that assembly theory aims to develop.
This prevents a (n equal-size) comparison between a biased sampling and another sampling in which the former was taken under distinct conditions in the latter.
In this case, without normalising by the summation $ \sum_{ i = 1 }^{ N } e^{ a_i } $, one could have the first one scoring a higher assembly number than the second one, but only because of the set of unique objects, and not because of the distribution of higher-assembly-index objects with respect to lower-assembly-index objects, which was what Assembly Theory intended in the first place.\footnote{In order to avoid the presence of false positives due to abiotic mimicry 
\cite{Gillen2023CallNewDefinitionBiosig,Malaterre2023ThereSuchThingBiosig}, an even better control group for excluding sampling condition biases in experiments could be achieved by only comparing ensembles with not only equal size and equal total summation of assembly indexes but also the same set of unique objects. Notice that the forthcoming results in Theorem~\ref{thmEncodingfromAN} with regards to assembly number hold analogously 
for any further normalisation than the one already given by Equation~\eqref{eqRenormAN}. }
For example, to tackle this issue, the value given by the measure
\begin{equation}\label{eqRenormAN}
	A'\#\left( X \right)
	=
	\frac{ \sum\limits_{ i = 1 }^{ N } e^{ a_i } \left( \frac{ n_i - 1 }{ N_T } \right) }{
	\sum\limits_{ \tiny \begin{array}{c} 
		X_j \in \mathbf{X}  \\
		\sum\limits_{ i = 1 }^{ N } e^{ a_i } = \sum\limits_{ i = 1 }^{ N_j } e^{ a_{ i j } }
	\end{array} }
	\; \left( 
	\sum\limits_{ i = 1 }^{ N_j } e^{ a_{ i j } } \left( \frac{ n_{ i j } - 1 }{ { N_T }_j } \right) 
	\right)}
\end{equation} 
improves on the originally proposed measure $ A\#\left( X \right) $ in \cite{croninnature}.

In case the set of possible objects is finite and the maximum size of the ensembles investigated is also finite, then there is a finite set $ \mathbf{X} $ of possible unique ensembles.
In this scenario, one can always keep effecting further renormalisations upon Equation~\eqref{eqRenormAN} to account for any possible ensemble in $ \mathbf{X} $, in this manner addressing a wider range of control groups in their experiments.

In general, for any other measure $ { A'' }\#\left( \cdot \right) $ that one can possibly devise such that this measure is monotonically dependent\footnote{ That is, if $ A\#\left( X_i \right) \leq A\#\left( X_j \right) $, then $ { A'' }\#\left( X_i \right) \leq { A'' }\#\left( X_j \right) $.} on $  A\#\left( \cdot \right) $ (e.g., by applying any form of normalisation), one can directly apply the Krafts inequalities and the source coding theorem \cite{Cover2005}, should the ensembles be indeed sampled accordingly and the number of samples sufficiently large:
one constructs an optimal code that statistically compresses each ensemble $ X $ in sequences whose length is $ C'_A\left( X \right) $ such that
\begin{equation}\label{eqShannonAN}
	- \log\left( { A'' }\#\left( X \right) \right)
	\leq
	C'_A\left( X \right)
	\leq
	- \log\left( { A'' }\#\left( X \right) \right)
	+
	\mathbf{O}(1)
\end{equation}
is expected to hold on average.
As proposed by Assembly Theory, notice from Equation~\eqref{eqShannonAN} that ensembles with higher assembly numbers are expected to be more compressible, and therefore they diverge more from those ensembles that are more statistically random.
This demonstrates that in the case of pure \emph{stochastic processes}, assembly number is either a \emph{suboptimal} or \emph{an optimal} compression method.

In case the set of possible ensembles can be arbitrarily large and the ensembles are \emph{not} generated by stationary and ergodic sources---for example, in this case the minimum expected codeword length per symbol may not converge to the entropy rate---finding the optimal code is not as straightforward as finding the code that generated Equation~\eqref{eqShannonAN}.
To tackle this generalised scenario, algorithmic complexity theory improves on any entropy-based coding methods, or on any pure statistical compression method.
By using algorithmic information theory (AIT), algorithmic complexity being one of its indexes, one obtains the universally optimal encoding,
and hence one demonstrates in Theorem~\ref{thmEncodingfromAN} and Corollary~\ref{thmMorecompressfromAN} that the assembly number value is equivalent to the degree of compressibility of the ensemble.
That is, the assembly number calculation method is subsumed into the universal coding of AIT, demonstrating that the assembly number, in fact, captures the idea of compressiblity of the ensemble---not to be conflated with the object--- although only as an approximation that can, in general, be arbitrarily distorted (see Section~\ref{sectionDistortionsforAN}). 

\begin{theorem}\label{thmEncodingfromAN}
	Let $ \mathcal{ S } $ and $ \mathbf{X} $ be enumerable by an algorithm.
	Let $ { A'' }\#\left( \cdot \right) $ be an arbitrary measure (or semimeasure) that can be effectively calculated by an algorithm 
	from the assembly number $ A\#\left( \cdot \right) $.
	Let $ \mathbf{F}' $ be an arbitrary formal theory that contains Assembly Theory, including all the decidable procedures of
	the chosen method for calculating the assembly number $ A\# $, the measure $ { A'' }\#\left( \cdot \right) $ for any subspace of $ \mathcal{ S } $ and ensemble in $ \mathbf{X} $,
	and the program that decides whether or not the criteria for building the assembly spaces or ensembles of objects are met.
	Let  $ X \in \mathbf{X} $ be an arbitrary ensemble.
	Then, one can encode/compress the ensemble $ X $ in $ C_A\left( X \right) $ bits such that
	\begin{equation}\label{eqCompressfromAN}
		- \log_{ 2 }\left( { A'' }\#\left( X \right) \right)
		\leq
		C_A\left( X \right) 
		\leq 
		- \log_{ 2 }\left( { A'' }\#\left( X \right) \right)
		+
		\mathbf{O}( 1 )
	\end{equation}
	and 
	\begin{equation}\label{eqLowerboundingAN}
			\mathbf{K}\left( X \right)
			\leq
			C_A\left( X \right)
			+ \mathbf{O}(1)
		\end{equation}
	hold, where $ C_A\left( X \right) $ denotes the size of the codeword/program for the ensemble $ X $. 

	\begin{proof}
		By hypothesis we have it that there is an algorithm that can always calculate the value of $ { A'' }\#\left( \cdot \right) $ for any $ X $.
		Hence, there will be another algorithm that can always calculate
		\begin{equation}\label{eq2EncodingfromAN}
			l_X \coloneqq  \lceil - \log_{ 2 }\left( { A'' }\#\left( X \right) \right) \rceil 
		\end{equation}
		from any ensemble $ X $.
		Then, one directly obtains that the sequence
		\begin{equation}
			\left( \left( l_{ X_1 } , X_1 \right) , \dots , \left( l_{ X_i } , X_i \right) , \dots \right)
		\end{equation}
		is enumerable\footnote{ In fact, not only computably enumerable, but in this case also computable.} by an algorithm.
		In addition, due to the fact that $ { A'' }\#\left( \cdot \right) $ is a (semi)measure by hypothesis, one has it that
		\begin{equation}
			\sum\limits_{ X_i \in \mathbf{ X } } \frac{ 1 }{ 2^{ l_{ X_i }  } } \leq 1
			\text{ .}
		\end{equation}
		From AIT~\cite{Downey2010}, one knows that if an arbitrary sequence $ \left( \left( d_1 , \tau_1 \right) , \dots , \left( d_i , \tau_i \right) , \dots \right) $ of pairs is enumerable and $ \sum\limits_{ i } 2^{ - d_i  } \leq 1 $, then one can encode each $ \tau_i $ in a codeword of length $ d_i $.
		Additionally, one also has it that
		\begin{equation}
			\mathbf{K}\left( \tau_i \right)
			\leq
			d_i
			+ \mathbf{O}(1)
			\text{ .}
		\end{equation}
		Therefore, from basic properties and inequalities in AIT, we will have it that there is an algorithm that can compress $ X $ in a program of length $ C_A\left( X \right) $ such that Equations~\eqref{eqCompressfromAN} and~\eqref{eqLowerboundingAN}
		hold.
	\end{proof}

\end{theorem}

Now, from Theorem~\ref{thmEncodingfromAN}, one can employ a monotone dependence on the assembly number to demonstrate that a higher assembly number indeed implies greater compressibility of the ensemble.

\begin{corollary}\label{thmMorecompressfromAN}
	Let the conditions for Theorem~\ref{thmEncodingfromAN} be met such that $ { A'' }\#\left( \cdot \right) $ is also \emph{strictly monotonic} on the assembly number $ A\#\left( \cdot \right) $.
	Then, there is $ k_0 \in \mathbb{ N } $ such that, if $ X , X' \in \mathbf{ X } $ are ensembles with $ A\#\left( X \right) + k_0 \leq A\#\left( X' \right) $, one has it that
	\begin{equation}\label{eqMorecompressfromAN}
		C_A\left( X' \right)
		<
		C_A\left( X \right)
	\end{equation}
	holds,
	where $ C_A\left( Y \right) $ denotes the size of the codeword/program for an arbitrary ensemble $ Y $.
	
	\begin{proof}
		Since  $ { A'' }\#\left( \cdot \right) $ is \emph{strictly} monotonic on $ A\#\left( \cdot \right) $,
		then for every $ \epsilon > 0 $, there is a sufficiently large $ \epsilon' > 0 $ such that, if $ A\#\left( X \right) + \epsilon' \leq A\#\left( X' \right) $, one has it that $ { A'' }\#\left( X \right) + \epsilon \leq { A'' }\#\left( X' \right) $ holds.
		As a consequence, one can always pick a sufficiently large $ k_0 \in \mathbb{ N } $ and arbitrary ensembles $ X $ and $ X' $ where $ A\#\left( X \right) + k_0 \leq A\#\left( X' \right)  $ such that
		\begin{equation}
			- \log_{ 2 }\left( { A'' }\#\left( X' \right) \right)
			+
			\mathbf{O}( 1 )
			<
			- \log_{ 2 }\left( { A'' }\#\left( X \right) \right)
		\end{equation}
		holds.
		Therefore, Equation~\eqref{eqMorecompressfromAN} follows from the inequality in Equation~\eqref{eqCompressfromAN}.
	\end{proof}
\end{corollary}

As proposed by \cite{cronin,marshall_murray_cronin_2017,Marshall2019,croninnature,croninentropy}, our results demonstrate that ensembles\footnote{Not to be conflated with the objects themselves.} with higher assembly numbers\footnote{Not to be conflated with the assembly indexes.} are more compressible, and therefore they diverge more from those ensembles that are outcomes of perfectly random processes.
A more random generative process for an ensemble (for which the assembly number scores lower) would mean that there are fewer constraints or biases playing a role in this process.
In turn, high-assembly-index objects would more rarely appear (or be constructed) because they are more complex (i.e., as we have demonstrated in the previous sections, less compressible or more incompressible)---this is what occurs in scenarios in which there is a weaker presence of top-down (or downward) causation \cite{Abrahao2021bEmergenceAIDPTRSA} behind the possibilities or paths that lead to the construction of the objects.
As demonstrated in this section, if one assumes that the assembly number indeed quantifies the presence of constraints or biases in the underlying generative processes of the ensembles, then a more compressible ensemble implies that more constraints or biases played a role in generating more unlikely (i.e., less compressible) objects more frequently than would have been the case in an environment with fewer constraints and biases (i.e., a more random or incompressible environment), hence increasing the frequency of occurrence of high-assembly-index (i.e., less compressible) objects in this environment.
Conversely, under the same assumption, if more biotic (or less random abiotic) processes being present in an ensemble imply a higher assembly number, then more biotic processes being present in an ensemble would also imply that the ensemble is more compressible.

Therefore, along with Section~\ref{sectionoAToverestimation}, the results in this
Section~\ref{sectionEncodingfromAN} demonstrate that
in the \emph{general case} (even when the underlying generative process of the ensembles is unknown) \emph{assembly number} $ A\#\left( \cdot \right) $ is an \emph{approximation} to algorithmic complexity (or algorithmic probability), but a \emph{suboptimal} measure with respect to the more general methods from algorithmic information theory.




\subsection{Lack of empirical evidence in favour of Assembly Theory and presence of empirical evidence against it}\label{sectionLogicalDepth}

The experiments and analysis of organic versus non-organic molecules from mass spectral data using several complexity indexes used as positive evidence for Assembly Theory in~\cite{marshall_murray_cronin_2017} empirically demonstrates what we have proven: that the assembly index is equivalent to statistical compression, given that it was not compared to any other data type or any other existing index of the same nature. The simplicity of the experiments required to prove the utility of the assembly index and the main hypothesis of Assembly Theory inclines us to believe that this was a serious omission.

Together with the theoretical results, Figure~\ref{figureBiologicalComplexityplot} from~\cite{salient} shows that the evidence that the AT authors present as being in favour of their theory and algorithms, specifically that their assembly (compression) index may correspond to the way objects assemble, is flawed and invalid. The empirical evidence actually shows that other, equally plausible and valid compression schemes that are sequential, count copies, and build a repeating symbol dictionary tree, are likely the ones from which AT's discriminatory power derives and on which AT's concepts and methods are based.

\begin{figure}[ht!]
	\centerline{\includegraphics[scale=0.8]{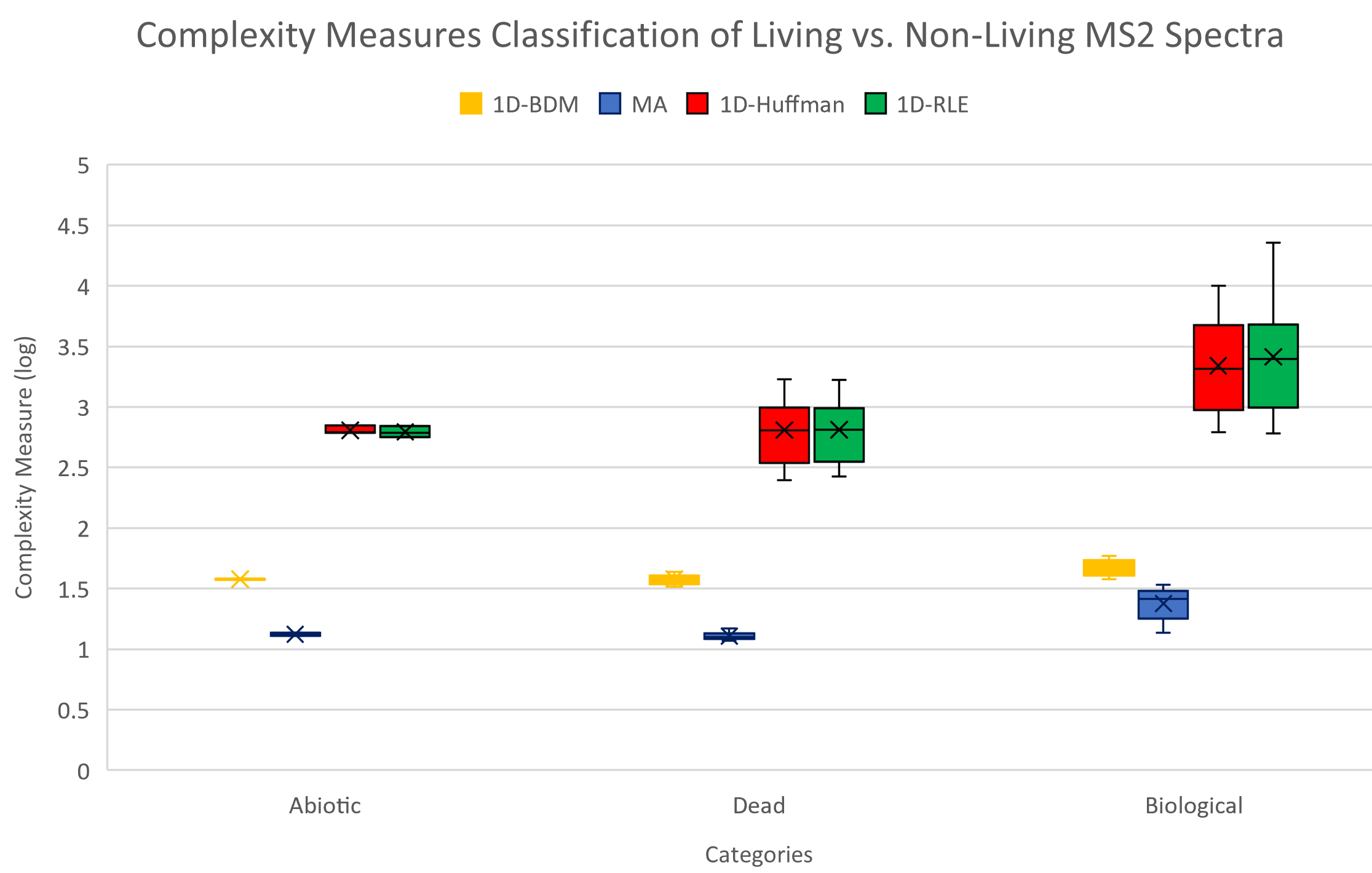}}
	\caption{\label{figureBiologicalComplexityplot}Taken from~\cite{salient}. The most basic statistical complexity indexes applied to the 18 molecular compounds for which the spectral data was made available from~\cite{marshall_murray_cronin_2017}. MA stands for Molecular Assembly based on the assembly index of the molecular compounds. These results were obtained even without optimising the algorithms after a simple flattening and bnarisation procedure of the data (full methods, data and code available at~\cite{salient}). The results show that other statistical measures separate the data just as MA does (or better) between living and non-living compounds.}
\end{figure}

The strongest positive correlation in Fig.~\ref{figureBiologicalComplexityplot}, was identified between MA and 1D-RLE coding (R= 0.9), which is one of the most basic coding schemes in computer science, known since the 60s, and among the most similar to the intended definition of MA, as being capable of `counting copies', together with LZ, which is exactly equivalent to MA or the assembly index. All indexes were applied to spectral data (only 18 extracts available in the original paper in~\cite{marshall_murray_cronin_2017}) and taking their results at face value. Other coding algorithms, including the Huffman coding (R = 0.896), also show a strong positive correlation with MA. The analysis further confirms our previous findings of the similarity in performance between MA and other basic compression measures (that essentially only `count copies' in data, all of them converging to Shannon Entropy) in classifying living vs. non-living mass spectra signatures. All molecules included in the Assembly Theory paper were included. In other words, both theory and experiment converge, and the experiments conform with the findings in this paper that Assembly Theory is equivalent to Shannon Entropy and their assembly index is an LZ compression algorithm renamed.


\end{appendices}

\end{document}